\documentclass[preprint]{elsarticle}

\usepackage{todonotes}
\usepackage{newfile}
\usepackage{totcount}
\usepackage{enumitem}
\usepackage{amssymb,amsmath,amsthm,graphicx,xcolor}
\usepackage[hidelinks]{hyperref}
\usepackage[capitalize]{cleveref}
\usepackage{subcaption}
\usepackage{tikz}
\usetikzlibrary{automata,arrows,decorations.pathreplacing,snakes,shapes.misc,graphs}
\tikzset{elliptic state/.style={draw,rounded rectangle,minimum height=1.25em}}
\tikzset{every initial by arrow/.style={initial text={ }}}
\tikzset{
  nfa/.style={
    ->,
    >=latex',
    shorten >=1pt,
    node distance=2.25cm,
    every state/.style={
      minimum size=1.25em
    },
    auto
  }
}

\newtheorem{theorem}{Theorem}[section]
\newtheorem{lemma}[theorem]{Lemma}
\newtheorem{proposition}[theorem]{Proposition}

\newtheorem*{claim*}{Claim}

\newproof{claimproof}{Proof of the claim}
\newproof{proofidea}{Proof Idea}

\newdefinition{definition}[theorem]{Definition}
\newdefinition{example}[theorem]{Example}
\newdefinition{remark}[theorem]{Remark}

\newcommand{\nat}[1][]{\eta_{}}
\newcommand{\nati}[1][]{\inv{\eta}_{}}
\newcommand{\lnf}{\operatorname{lnf}}
\usepackage[only,llbracket,rrbracket]{stmaryrd}

\makeatletter
\newcommand{\mylabel}[2]{#2\def\@currentlabel{#2}\label{#1}}
\makeatother

\newcommand{\INC}{\operatorname{inc}}
\newcommand{\DEC}{\operatorname{dec}}
\newcommand{\twinindex}{\operatorname{TI}}
\newcommand{\settwinindex}{\twinindex_{*}}
\newcommand{\inv}[1]{#1^{-1}}

\newcommand{\traceMon}[1][\structD]{\monoM(#1)}

\newcommand{\ov}{\overline}
\newcommand{\alphabet}[1]{\mathrm{Alph}(#1)}

\newcommand{\automatonfont}[1]{\mathfrak{#1}}
\newcommand{\monoidfont}[1]{\mathbb{#1}}
\newcommand{\structurefont}[1]{\mathcal{#1}}

\makeatletter
\newcommand{\@abbrev}[3]{
  \def\c@a@def##1{
    \if ##1.
    \relax
    \else
    \@ifdefinable{\@nameuse{#1##1}}{\@namedef{#1##1}{#2##1}}
    \expandafter\c@a@def
    \fi
  }
  \c@a@def #3.
}

\@abbrev{bb}{\mathbb}{ABCDEFGHIJKLMNOPQRSTUVWXYZ}
\@abbrev{bf}{\mathbf}{ABCDEFGHIJKLMNOPQRSTUVWXYZabcdefghijklmnopqrstuvwxyz}
\@abbrev{bit}{\boldsymbol}{ABCDEFGHIJKLMNOPQRSTUVWXYZabcdefghijklmnopqrstuvwxyz}
\@abbrev{cal}{\mathcal}{ABCDEFGHIJKLMNOPQRSTUVWXYZ}
\@abbrev{frak}{\mathfrak}{ABCDEFGHIJKLMNOPQRSTUVWXYZabcdefghijklmnopqrstuvwxyz}
\@abbrev{up}{\mathrm}{ABCDEFGHIJKLMNOPQRSTUVWXYZabcdefghijklmnopqrstuvwxyz}
\@abbrev{scr}{\mathscr}{ABCDEFGHIJKLMNOPQRSTUVWXYZ}
\@abbrev{sf}{\mathsf}{ABCDEFGHIJKLMNOPQRSTUVWXYZabcdefghijklmnopqrstuvwxyz}
\@abbrev{ov}{\ov}{ABCDEFGHIJKLMNOPQRSTUVWXYZabcdefghijklmnopqrstuvwxyz}
\@abbrev{struct}{\structurefont}{ABCDEFGHIJKLMNOPQRSTUVWXYZ}
\@abbrev{aut}{\automatonfont}{ABCDEFGHIJKLMNOPQRSTUVWXYZ}
\@abbrev{mono}{\monoidfont}{ABCDEFGHIJKLMNOPQRSTUVWXYZ}
\@abbrev{graph}{\structurefont}{ABCDEFGHIJKLMNOPQRSTUVWXYZ}
\makeatother

\newcommand{\Conf}[1]{\mathrm{Conf}_{#1}}

\DeclareMathOperator{\twins}{\mathrm{twns}}

\definecolor{darkred}{HTML}{CC0000}

\newcommand{\Id}{\mathrm{Id}}

\newcommand{\relation}[1]{\mathrel{#1}}

\newcommand{\reach}{\operatorname{Reach}}

\usepackage{comment}

\begin{document}
\title{Reachability in Trace-Pushdown Systems}

\author[1]{Chris Köcher\fnref{fn1}}
\ead{ckoecher@mpi-sws.org}

\author[2]{Dietrich Kuske}
\ead{dietrich.kuske@tu-ilmenau.de}

\fntext[fn1]{This work was done while Chris Köcher was affiliated with the Technische Universität Ilmenau.}

\affiliation[1]{
  organization={Max Planck Institute for Software Systems},
  addressline={Paul-Ehrlich-Straße 26},
  postcode={67663},
  city={Kaiserslautern},
  country={Germany}
}

\affiliation[2]{
  organization={Technische Universität Ilmenau},
  addressline={Ehrenbergstraße 29},
  postcode={98693},
  city={Ilmenau},
  country={Germany}
}

\begin{abstract}
  We consider the reachability relation of pushdown systems whose
  pushdown holds a Mazurkiewicz trace instead of just a word as in
  classical systems.  Under two natural conditions on the transition
  structure of such systems, we prove that the reachability relation
  is lc-rational, a new notion that restricts the class of rational
  trace relations. We also develop the theory of these lc-rational
  relations to the point where they allow to infer that
  forwards-reachability of a trace-pushdown system preserves the
  rationality and backwards-reachability the recognizability of sets
  of configurations. As a consequence, we obtain that it is decidable
  whether one recognizable set of configurations can be reached from
  some rational set of configurations. All our constructions
  are polynomial (assuming the dependence alphabet to be fixed).
  
  These findings generalize results by Caucal on classical pushdown
  systems (namely the rationality of the reachability relation of such
  systems) and complement results by Zetzsche (namely the
  decidability for arbitrary transition structures under severe
  restrictions on the dependence alphabet).
\end{abstract}
\begin{keyword}
  Reachability \sep Formal Verification \sep Pushdown Automaton \sep Distributed System
\end{keyword}

\newpageafter{abstract}

\maketitle

\section{Introduction}

Much work has been done on the verification of pushdown systems that
model the behavior of sequential recursive programs. The most
fundamental result in this direction is the efficient decidability of
the reachability relation. This result was first shown by Caucal
\cite{Cau88} when he proved the reachability relation to be prefix
recognizable. To this aim, he modified the pushdown system in such a
way that the reachability relation does not change, but any path
connecting two configurations first shortens the pushdown and, in the
second phase, does not shorten it in any step. Since prefix
recognizable relations are rational, it follows that the sets of
configurations reachable forwards or backwards, respectively, from
some regular set of configurations is regular, again. One therefore
gets that even the question ``is it possible to reach some
configuration from this regular set starting in some configuration
from that regular set'' is decidable. The preservation of regularity
under backwards reachability was later reproved using a completely
different technique by Esparza et at.~\cite{BouEM97}. Differently from
Caucal \cite{Cau88}, they directly modified an NFA accepting a set of
configurations in such a way that the resulting NFA accepts the set of
configurations backwards reachable from the former set; their
technique does not immediately imply Caucal's preservation result for
the forwards reachability \cite{Sch19}.

Pushdown systems are a special case of valence automata
\cite{IbaSK76,Gil96,ItoMVM01,MitS01,ElsO04,EldKO08,Kam09,REnK09,Ren10,Zet15,Zet21,OsuMZ16}:
in a pushdown system, the possible pushdown contents come from a free
monoid and can be accessed at the prefix, only. Differently, in a
valence automaton, this free monoid is replaced by an arbitrary
monoid. Zetzsche \cite{Zet15,Zet21,OsuMZ16} considered the question
what monoids yield decidable reachability relations in valence
automata. To this aim, he considers graph monoids \cite{Zet13} that
allow to model pushdown systems, Petri nets, multi-stack automata,
counter automata and many more. He describes a large class of graph
monoids that result in decidable reachability relations, and another
class that allows to construct valence automata with an undecidable
reachability relation.

In this paper, we consider pushdown systems that hold, in their
pushdown, not a word, but a Mazurkiewicz trace. These models can,
alternatively, be understood as valence automata over loop-free graph
monoids. We aim at results similar to those by Caucal, i.e.,
preservation results under forwards and backwards reachability. At the
same time, our approach complements Zetzsche's work: while he was
concerned with properties of the graph monoid, we describe properties
of the transition structure of the automaton that guarantee
decidability of the reachability relation.

As we aim at an extension of Caucal's result, we ask which properties
of a set of configurations are preserved under the reachability. In
the case of pushdown systems, this is the case regarding
regularity. In our systems, the set of configurations is a set of
Mazurkiewicz traces where there are two distinct generalizations of
regularity: rationality and recognizability. Our main result
(Theorem~\ref{thm:reach}) states that
\begin{enumerate}[label=(\arabic*)]
\item recognizability is preserved under backwards reachability, but
  not under forwards reachability and
\item rationality is preserved under forwards reachability, but not
  under backwards reachability.
\end{enumerate}
To obtain the preservation results for pushdown systems, Caucal first
proved the rationality of the reachability relation and then employed
that regularity is preserved under rational relations. This
preservation does not hold for rational trace relations
(Lemma~\ref{lem:nonPreservation}). Therefore, we introduce a new
notion that we call left-closed rational trace relations (short:
lc-rational relations) and develop their theory as far as it is needed
in our context. Namely, we show that
\begin{enumerate}[label=(\roman*)]
\item a suitable generalization of Caucal's prefix-recognizable
  relations to the trace setting is lc-rational (Theorem~\ref{thm:product}),
\item the composition and union of lc-rational relations is
  lc-rational (Proposition~\ref{prop:composition}),
\item and the left-application of lc-rational relations preserves
  the rationality and the right-application preserves the
  recognizability of sets of traces (Theorem~\ref{thm:preservation}).
\end{enumerate}

Coming back to trace-pushdown systems, we then prove that the
reachability relation of a trace-pushdown system is
lc-rational. Recall that, following Caucal's ideas, a pushdown system
can be modified such that any path can be simulated by a path
consisting of at most two phases: in the first one, the pushdown is
shortened and, in the second, no shortening takes place. While this is
not the case for trace-pushdown systems, we are able to show that a
uniformly bounded number of phases suffices (the bound only depends on
the trace monoid, Proposition~\ref{prop:swallow}). Using (i), we also
show that the restriction of the reachability relation to single-phase
paths is lc-rational (Propositions~\ref{prop:read}
and~\ref{prop:write}). Hence (ii) implies that the reachability
relation of a trace-pushdown system is lc-rational
(Theorem~\ref{thm:reachLCRational}). Now the preservation results from
(1) and (2) follow immediately from (iii).

As all our constructions are efficient, it follows that it is
decidable in polynomial time whether a recognizable set of traces is
forwards reachable from some rational (e.g., recognizable) set of
traces. In that sense, we generalize Caucal's result to the realm of
trace pushdown systems.

\section{Preliminaries}\label{sec:prelim}

In this paper, $0$ is a natural number and, for $n\in\monoN$, we set
$[n]=\{1,2,\dots,n\}$.

Let $A$ be some alphabet and $w\in A^*$ a word over $A$. The
\emph{alphabet of $w$}, denoted $\alphabet{w}\subseteq A$, is the set
of letters occuring in the word $w$. Furthermore, $|w|_a$ denotes the
number of occurrences of the letter $a\in A$ in the word $w$.

\subparagraph{Trace Theory} A \emph{dependence alphabet} is a pair
$\structD=(A,D)$ where $A$ is a finite set of \emph{letters} and
$D\subseteq A\times A$ is a reflexive and symmetric relation, the
\emph{dependence relation}. For a letter $a\in A$, we write $D(a)$ for
the set\linebreak $\{b\in A \mid (a,b)\in D\}$ of letters dependent from $a$,
$D(B)=\bigcup_{b\in B}D(b)$ is the set of letters dependent from some
letter in $B\subseteq A$. Note that $a\in D(a)$ since $D$ is reflexive
and $a\in D(b)$ iff $b\in D(a)$ since $D$ is symmetric. For a word
$w\in A^*$, let $D(w)=D(\alphabet{w})$ be the set of letters dependent
from some letter in the word $w$.  The \emph{independence relation}
$I\subseteq A\times A$ is the set of pairs $(a,b)$ of distinct letters
with $(a,b)\notin D$.  If $\alphabet{u}\times\alphabet{v}\subseteq I$
for two words $u,v\in A^*$, i.e., if any letter $a$ from $u$ is
independent from any letter $b$ from $v$, then we write
$u\parallel v$.

To measure the complexity of our algorithms, we will need the
following parameters of a dependence alphabet $\structD=(A,D)$:
\begin{itemize}
\item The \emph{size} $\|\structD\|$ is the number $|A|$ of letters.
\item Two letters $a,b\in A$ are \emph{twins} if they have the same
  dependent letters, i.e., if $D(a)=D(b)$ (implying in particular
  $(a,b)\in D$ since $a\in D(a)$). We denote by
  $\twins(a)=\{b\in A\mid D(a)=D(b)\}$ the set of all twins of the
  letter $a$.

  By $\twins(\structD)=\{\twins(a)\mid a\in A\}$ we denote the set of
  all equivalence classes of the relation ``twin''; the \emph{twin
    index} $\twinindex(\structD)=|\twins(\structD)|$ is the number of
  these equivalence classes. Note that $\twinindex(\structD)$ equals
  the number of sets $D(a)$ for $a\in A$.
\item The \emph{set twin index} $\settwinindex(\structD)$ is the number of sets
  $D(B)$ for $B\subseteq A$ (for $B,C\subseteq A$, we have $D(B)=D(C)$
  iff $B$ and $C$ are twins in the graph $(A,D)$; the set twin index
  is the number of equivalence classes of the relation ``set twin'').
\item The \emph{independence number} $\alpha(\structD)$ is the maximal
  size of a set $B\subseteq A$ of mutually independent letters (the
  notation $\alpha(\structD)$ is standard in graph theory).
\end{itemize}
Since independent letters cannot be twins, we get
$\alpha(\structD)\le\twinindex(\structD)$; since any set $D(B)$ equals
a union of sets $D(a)$, we have
$\twinindex(\structD) \le \settwinindex(\structD) \le
2^{\twinindex(\structD)}$. 

Let ${\sim}\subseteq A^*\times A^*$ denote the least monoid congruence
with $ab\sim ba$ for all $(a,b)\in I$. In other words, $u\sim v$ holds
for two words $u,v\in A^*$ iff $u$ can be obtained from $v$ by
transposing consecutive independent letters. In particular, $u\sim v$
implies $|u|=|v|$ as well as $\alphabet{u}=\alphabet{v}$. Furthermore,
$u\parallel v$ implies $uv\sim vu$ (but the converse implication does
not hold as the example $u=v=ab$ shows).%\todo{$(a,b)\in D$ ist unnötig, da $I$ irreflexiv ist. Die Aussage gilt für alle $u=v\neq\varepsilon$.}

The \emph{(Mazurkiewicz) trace monoid} induced by $\structD$ is the
quotient of the free monoid $A^*$ wrt.\ the congruence $\sim$, i.e.,
$\traceMon=A^*/{\sim}$. Its elements are equivalence classes of words
denoted $[w]$; by $[w]$, we mean the equivalence class containing $w$,
it is the \emph{trace induced by $w$}.

Suppose $D=A\times A$, i.e., all letters are mutually dependent. Then
$u\sim v\iff u=v$ holds for all words $u,v\in A^*$; hence
$\traceMon\cong A^*$ in this case. Furthermore,
$\twinindex(\structD)=\settwinindex(\structD)=\alpha(\structD)=1$.

The other extreme is $D=\{(a,a)\mid a\in A\}$ where any two distinct
letters are independent. Then $u\sim v$ iff $|u|_a=|v|_a$ holds for
any letter $a\in A$. Hence $\traceMon\cong(\bbN^{|A|},+)$ as well as
 $\twinindex(\structD)=\alpha(\structD)=|A|$ and
$\settwinindex(\structD)=2^A$ hold in this case.

For a set $B\subseteq A$ of letters, let $\pi_B\colon A^*\to B^*$
denote the monoid homomorphism with $\pi_B(b)=b$ for $b\in B$ and
$\pi_B(a)=\varepsilon$ for $a\in A\setminus B$. We refer to this
mapping as the \emph{projection to $B$}. Now let $B\subseteq A$ be a
set of pairwise dependent letters (i.e., $B\times B\subseteq D$). Then
the very definition of $\sim$ ensures $\pi_B(u)=\pi_B(v)$ for any
words $u$ and $v$ with $u\sim v$. 
\begin{theorem}[cf.~\protect{\cite[Cor.~1.4.5]{Die90}}]
  Let $\structD=(A,D)$ be a dependence alphabet and
  $B_1,\dots,B_n\subseteq A$ sets of letters with
  $D=\bigcup_{1\le i\le n}B_i\times B_i$. For any words $u,v\in A^*$,
  $u\sim v$ if, and only if, $\pi_{B_i}(u)=\pi_{B_i}(v)$ for all
  $1\le i\le n$.
\end{theorem}

It follows that the trace monoid $\traceMon$ is isomorphic to the
submonoid of $\prod_{(a,b)\in D}\{a,b\}^*$ generated by all tuples
$(\pi_{\{a,b\}}(c))_{(a,b)\in D}$ for $c\in A$. As a consequence, the
trace monoid is cancellative, i.e., $s\cdot t\cdot u=s\cdot t'\cdot u$
implies $t=t'$ for any traces $s,t,t',u\in\traceMon$. We also get that
$[a] \cdot s=[b]\cdot t$ with $(a,b)\in D$ implies $a=b$ and therefore
$s=t$ (for any $a,b\in A$ and $s,t\in\traceMon$).

For a comprehensive survey of trace theory see \cite{DieR95}.

\subparagraph{Automata and Word Languages} An
\emph{$\varepsilon$-NFA} or \emph{nondeterministic finite automaton
  with $\varepsilon$-transitions} is a tuple $\autA=(Q,A,I,T,F)$ where
$Q$ is a finite set of \emph{states}, $A$ is an alphabet,
$I,F\subseteq Q$ are the sets of \emph{initial} and \emph{final}
states, respectively, and
\[
  T\subseteq Q\times (A\cup\{\varepsilon\}) \times Q
\]
is the set of \emph{transitions}. Its size $\|\autA\|$ is defined to be
$|Q|+|A|$. The $\varepsilon$-NFA $\autA$ is an \emph{NFA} if
$T\subseteq Q\times A\times Q$; it is a \emph{deterministic finite
  automaton} or \emph{DFA} if, in addition, $I=\{\iota\}$ is a
singleton and, for any $(p,a)\in Q\times A$, there is a unique state
$q\in Q$ with $(p,a,q)\in T$.

Let $\autA=(Q,A,I,T,F)$ be an $\varepsilon$-NFA. A \emph{path} is a
sequence
\[
  (p_0,a_1,p_1)(p_1,a_2,p_2)\cdots (p_{n-1},a_n,p_n)
\]
of matching transitions (i.e., elements of $T$). Such a path is usually denoted
\[
  p_0 \xrightarrow{a_1} p_1
  \xrightarrow{a_2} p_2
  \cdots
  \xrightarrow{a_n} p_n
\]
or, if the intermediate states are of no importance,
\[
  p_0 \xrightarrow{a_1a_2\cdots a_n} p_n\,.
\]
This path is \emph{accepting} if it connects an initial state with a
final state, i.e., if $p_0\in I$ and $p_n\in F$. It accepts the word
$w=a_1a_2\cdots a_n$ (note that $a_i\in A\cup\{\varepsilon\}$ such
that this word $w$ can have length properly smaller than $n$). We
denote by $L(\autA)$ the set of words $w$ accepted by $\autA$. A
language $L\subseteq A^*$ is \emph{regular} if it is accepted by some
NFA, i.e., if there is some NFA $\autA$ with $L=L(\autA)$.

A foundational result in the theory of finite automata states that
$\varepsilon$-NFA, NFA, and DFA are equally expressive. Even more,
$\varepsilon$-NFA can be transformed into equivalent NFA in
polynomial time while the transformation of an NFA into an equivalent
DFA requires exponential time.

A language $L\subseteq A^*$ is \emph{rational} if it can be
constructed from finite languages using the operations union,
multiplication, and Kleene star. By Kleene's theorem \cite{Kle56}, a
language is regular if, and only if, it is rational.

\subparagraph{Transducer and Word Relations}

A \emph{transducer} is a quintuple $\autT=(Q,A,I,T,F)$ where $Q$, $A$,
$I$, and $F$ are as for $\varepsilon$-NFAs, and
\[
  T\subseteq Q \times (A^*\times A^*) \times Q
  \text{ with }
  ((p,(u,v),q)\in T\ \Rightarrow\ |uv|\le1)
\]
is a set of transitions that are labeled by pairs $(a,b)$. Note that
transitions are labeled by a pair consisting of a letter and the empty
word or of two empty words. The size $\|\autT\|$ of the transducer
$\autT$ is defined to be $|Q|+|A|$ which is fine since the number of
transitions is polynomial in this size measure. A \emph{path} is a
sequence
\[
  (p_0,(a_1,b_1),p_1)(p_1,(a_2,b_2),p_2)\cdots (p_{n-1}(a_n,b_n),p_n)
\]
of matching transitions. As in the case of $\varepsilon$-NFAs, we
usually denote it by
\[
  p_0\xrightarrow{(a_1,b_1)}p_1
  \xrightarrow{(a_2,b_2)}p_2
  \cdots\xrightarrow{(a_n,b_n)}p_n
  \text{ or }
  p_0\xrightarrow{(a_1a_2\cdots a_n,b_1b_2\cdots b_n)}p_n
\]
(note that the definition of a transducer requires $|a_ib_i|\le 1$,
hence the words $a_1a_2\cdots a_n$ and $b_1b_2\cdots b_n$ can have
different length, the sum of these lengths is at most $n$). This path
is \emph{accepting} if $p_0\in I$ and $p_n\in F$. A pair of words
$(u,v)\in A^*\times A^*$ is accepted by $\autT$ if there is a path
\[
  I\ni p\xrightarrow{(u,v)}p_n\in F\,.
\]
By $R(\autT)\subseteq A^*\times A^*$, we denote the set of pairs
$(u,v)$ that are accepted by $\autT$.

A word relation $R\subseteq A^*\times A^*$ is \emph{rational} if it
can be constructed from finite relations using the operations union,
multiplication, and Kleene star. A foundational result on word
relations (\cite{ElgM65}, cf.\ \cite[Thm.~III.6.1]{Ber79}) states that a
word relation $R$ is rational if, and only if, it is accepted by some
transducer, i.e., there is a transducer $\autT$ with $R(\autT)=R$.

\subparagraph{Rational and Recognizable Trace Languages}

Fix some dependence alphabet $\structD=(A,D)$.  For a word language
$L\subseteq A^*$, we denote by $[L]$ the set of traces $[u]$ induced
by words from $L$, i.e., $[L]=\{[u]\mid u\in L\}\subseteq\traceMon$.

Now let $\calL\subseteq\traceMon$. The set $\calL$ is
\emph{recognizable} if the word language
$\{u\in A^*\mid [u]\in \calL\}$ is regular; it is \emph{rational} if
there exists a regular word language $L\subseteq A^*$ with
$\calL=[L]$. It follows that every recognizable trace
language is rational; the converse implication is known to fail
(consider, e.g., the trace language $\{[ab]^n\mid n\in\bbN\}$ with
$(a,b)\notin D$ that is rational, but not recognizable).

Since rational trace languages are images of rational word languages,
we obtain that $\calL\subseteq\traceMon$ is rational if, and only if,
it can be constructed from finite trace languages using the operations
union, multiplication, and Kleene star.

By the very definition, every NFA $\autA$ over the alphabet $A$
represents a rational trace language
$[L(\autA)]=\{[u]\mid u\in L(\autA)\}$. A word language
$L\subseteq A^*$ is \emph{closed} if $u\sim v\in L$ implies $u\in L$;
we call an NFA \emph{closed} if its language $L(\autA)$ is closed. In
this case, the trace language $[L(\autA)]$ is even recognizable. Note
that $\{u\in A^*\mid [u]\in \calL\}$ is closed for any trace language
$\calL$. Hence every recognizable trace language $\calL$ can be
represented by some closed NFA $\autA$. 

Even more, every recognizable trace language $\calL$ is represented by
some DFA that satisfies the following diamond property for any
$(a,b)\in I$ and $p\in Q$ (see \cref{fig:diamond-DFA} for a
vizualisation that should also explain the name ``diamond property''):
\begin{enumerate}[label=(D)]
\item For each $(p,a,q),(q,b,r)\in T$, there is a $q'\in Q$
  with $(p,b,q'),(q',a,r)\in T$.\label{def:diamond1-DFA}
\end{enumerate}

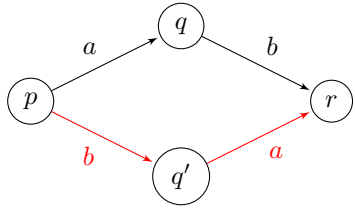
\begin{figure}[h]
  \begin{center}
    \begin{tikzpicture}[nfa]
      \node[state] (p) at (0,0) {$p$};
      \node[state] (q1) at (2,1) {$q$};
      \node[state] (q2) at (2,-1) {$q'$};
      \node[state] (r) at (4,0) {$r$};
      
      \draw (p) edge node {$a$} (q1)
                edge[color=red] node[below left] {$b$} (q2)
            (q1) edge node {$b$} (r)
            (q2) edge[color=red] node[below right] {$a$} (r);
    \end{tikzpicture}
  \end{center}
  \caption{Visualization of the diamond property
    \ref{def:diamond1-DFA} of an NFA. It states that whenever we find
    the black transitions with $a\parallel b$, we also find a
    state $q'$ with the red transitions. }
  \label{fig:diamond-DFA}
\end{figure}

\subparagraph{Rational Trace Relations}

Fix some dependence alphabet $\structD=(A,D)$.  For a word relation
$R\subseteq A^*\times A^*$, we denote by $[R]$ the set of pairs
$([u],[v])$ of traces  for  word pairs  from $R$, i.e.,
$[R]=\{([u],[v])\mid (u,v)\in R\}$.

A trace relation $\calR\subseteq\traceMon\times\traceMon$ is
\emph{rational} if there exists a rational word relation
$R\subseteq A^*\times A^*$ with
$\calR=\{([u],[v])\mid (u,v)\in R\}=[R]$. Hence, rational trace
relations can be represented by transducers.

Since the mapping $A^*\to\traceMon\colon u\mapsto[u]$ is a monoid
homomorphism, a trace relation $\calR$ is rational if, and only if, it
can be obtained from finite relations using the operations union,
multiplication, and iteration.

\subparagraph{Left and Right Application}
Let $X$ be a set, $L\subseteq X$ be a subset of $X$, and
$R\subseteq X\times X$ be a binary relation on $X$. Then we set
\[
  L^R=\{y\in X\mid \exists x\in L\colon(x,y)\in R\}\text{ and }
  {}^RL=\{x\in X\mid \exists y\in L\colon(x,y)\in R\}\,.
\]
If $R$ is (the graph of) a function $f\colon X\to X$, then
$L^R$ is the image of $L$ under $f$, i.e., $L^R=\{f(x)\mid x\in L\}$,
and ${}^RL$ is the preimage of $L$ under $f$, i.e.,
${}^RL=\{x\in X\mid f(x)\in L\}$. Often, authors write $L\,R$ for
$L^R$ and $R\,L$ for ${}^RL$; we prefer the above notation as it stresses
the different roles played by the set $L$ and the relation $R$.

\begin{example}
  Let $\structD=(A,D)$ be a dependence alphabet and $u\in A^*$. Then
  $D(u)=\{b\in A\mid \exists a\in\alphabet{u}\colon(a,b)\in
  D\}=\alphabet{u}^D={}^D\alphabet{u}$---nevertheless, we will stick
  with the notation $D(u)$ in this case.
\end{example}

The mapping $2^X\to2^X\colon L\mapsto L^R$ is the
\emph{right-application} of $R$ while the mapping
$2^X\to 2^X\colon L\mapsto{}^RL$ is the \emph{left-application} of
$R$.  Note that for $\inv{R}=\{(y,x)\mid(x,y)\in R\}$, we have
${}^RL=L^{\inv{R}}$ and, since $\inv{(\inv{R})}=R$, also
${}^{\inv{R}}L=L^R$.

\subparagraph{Convention} In this paper, we will regularly consider
subsets of and binary relations on $A^*$ and $\traceMon$, resp. We
hope to simplify understanding by using the following conventions:
\begin{itemize}
\item Subsets of $A^*$ are denoted by plain capital letters $K$ and
  $L$; binary relations on $A^*$ are similarly denoted $R$, $R_1$, and
  $R_2$.
\item Subsets of $\traceMon$ are denoted by curly letters $\calK$ and
  $\calL$; binary relations on $\traceMon$ are similarly denoted
  $\calR$, $\calR_1$, and $\calR_2$.
\end{itemize}

\section{Trace-Pushdown Systems and Problem Statement}

Recall that a \emph{pushdown system} \cite{Cau88,BouEM97,FinWW97} is a pair
$(Q,\Delta)$ where $Q$ is a finite set of \emph{states} and
$\Delta\subseteq Q\times A\times A^*\times Q$ is a finite set of
\emph{transitions}. Here, $(p,a,w,q)\in\Delta$ describes that the system can
move from state $p$ to state $q$ while replacing the letter $a$ at the
top of its pushdown by the word $w$. In particular, a pushdown holds,
in every configuration, a word $w\in A^*$ that can be accessed at the
prefix, only.

In this paper, we consider pushdowns that hold a trace
$[w]\in\traceMon$ which can, as before, be accessed at the prefix,
only. To this aim, we define the \emph{trace semantics} of a pushdown
system as follows. Fix a dependence alphabet $\structD=(A,D)$ and let
$\autP=(Q,\Delta)$ be a pushdown system. The set of
configurations $\Conf{\autP}$ of $\autP$ is $Q\times\traceMon$. For
two configurations $(p,[u]),(q,[v])\in\Conf{\autP}$, we set
$(p,[u])\vdash_\autP(q,[v])$ if there is a transition
$(p,a,w,q)\in\Delta$ and a word $x\in A^*$ such that $[u]=[ax]$ and
$[v]=[wx]$. Note that $[u]=[ax]$ is equivalent to saying
$[u]=[a]\cdot[x]$ and similarly $[v]=[wx]$ is equivalent to
$[v]=[w]\cdot[x]$. Hence $(p,s)\vdash_\autP(q,t)$ for $p,q\in Q$ and
$s,t\in\traceMon$ if there is a transition $(p,a,w,q)\in\Delta$ such
that the trace $t$ results from the trace $s$ by replacing the prefix
$[a]$ by $[w]$.

The reflexive and transitive closure of the one-step relation
$\vdash_\autP$ is the \emph{reachability relation}
$\vdash^*_\autP$. We write $\vdash$ and $\vdash^*$ instead of
$\vdash_\autP$ and $\vdash^*_\autP$ whenever the situation is clear.

The pushdown systems with trace semantics form a special case of
valence automata over graph monoids \cite{Zet21}. Zetzsche aimed at
properties of the dependence alphabet $\structD$ that ensure
decidability of the reachability relation. He obtained the following
two properties that are necessary and sufficient, respectively.

\begin{theorem}[Zetzsche~\relax{\cite{Zet21}}]
  Let $\structD=(A,D)$ be a dependence alphabet and $I=A^2\setminus D$
  the associated independence relation.
  \begin{itemize}
  \item If $(A,I)$ contains the cycle or the path on four vertices
    (i.e., $C_4$ or $P_4$) as an induced subgraph, then there exists a
    pushdown system with an undecidable reachability relation (in the
    trace semantics).
  \item If the independence relation $I$ is transitive, then the
    reachability relation (in the trace semantics) of every pushdown
    system is decidable.
  \end{itemize}
\end{theorem}
Note that the dependence alphabet $\structD=(A,D)$ with $A=\{a,b,c\}$,
$(a,b)\in D$, and $(a,c),(b,c)\in I$ is not covered by any of these
two conditions.

Differently from Zetzsche, we aim at properties of the transition
structure of the pushdown system that ensure decidability. These
properties are captured by the notion of a trace-pushdown system
defined as follows.

\begin{definition}\label{def:cpds}
  Let $\structD=(A,D)$ be a dependence alphabet. A
  \emph{trace-pushdown system} (or \emph{tPDS}, for short) is a pushdown
  system with trace semantics $\autP=(Q,\Delta)$ over $\structD$ such
  that the following hold:
  \begin{enumerate}[label=(P\arabic*)]
  \item for each $(p,a,w,q)\in\Delta$ we have
    $D(w)\subseteq D(a)$,\label{def:tPDS1}
  \item[\mylabel{def:diamond1'}{(P2')}]
    for each $(p,a,v,q),(q,b,w,r)\in\Delta$ with
    $av\parallel bw$, there is a state $q'\in Q$ with
    $(p,b,w,q'),(q',a,v,r)\in\Delta$. 
  \end{enumerate}

  The \emph{size} of $\autP=(Q,\Delta)$ is
  $\|\autP\|:=|Q|+|A|+k\cdot|\Delta|$ where $k-1$ is the maximal
  length of a word occurring in any transition of $\autP$ (i.e.,
  $\Delta\subseteq Q\times A\times A^{<k}\times Q$).
\end{definition}

\begin{figure}[h]
  \begin{center}
    \begin{tikzpicture}[nfa]
      \node[state] (p) at (0,0) {$p$};
      \node[state] (q1) at (2,1) {$q$};
      \node[state] (q2) at (2,-1) {$q'$};
      \node[state] (r) at (4,0) {$r$};
      
      \draw (p) edge node {$a\mid v$} (q1)
                edge[color=red] node[below left] {$b\mid w$} (q2)
            (q1) edge node {$b\mid w$} (r)
            (q2) edge[color=red] node[below right] {$a\mid v$} (r);
    \end{tikzpicture}
  \end{center}
  \caption{Visualization of the diamond property \ref{def:diamond1'} of
    a tPDS. It states that whenever we find the black
    transitions with $av\parallel bw$, we also find a state $q'$ with
    the red transitions. }
\end{figure}
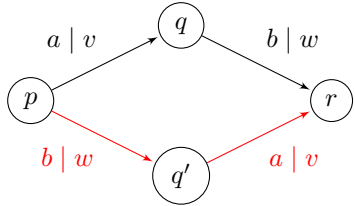

\begin{example}\label{ex:tPDS}
  Consider the dependence alphabet $\structD=(A,D)$ with
  $A=\{a,b,c\}$, $(a,c),(b,c)\in D$, and $(a,b)\notin D$. Let $\autP$
  be the pushdown system with $Q=\{p\}$ and
  $\Delta=\{(p,c,ca,p),(p,c,cab,p)\}$. It can easily be checked that
  $\autP$ is a trace-pushdown system since it satisfies
  \ref{def:tPDS1} and \ref{def:diamond1'}.

  Note that $cab\,ab\sim ca^2 b^2$ since $(a,b)\notin D$. Therefore,
  we have
  \begin{align*}
    (p,[c])&\vdash (p,[cab]) \vdash (p,[cab\,ab])=(p,[ca^2 b^2])\\
    &\vdash(p,[cab\,a^2 b^2])=(p,[ca^3b^3]) \vdash^*(p,[ca^n b^n])
  \end{align*}
  for all $n\ge3$.

  The configurations reachable from $(p,[c])$ together with the
  one-step relation between them are depicted in
  Fig.~\ref{fig:ex:tPDS}. Note that this configuration graph is
  isomorphic to the infinite grid. Hence its monadic second-order
  theory is undecidable and, consequently, there is no pushdown system
  with an isomorphic configuration graph.\qed

  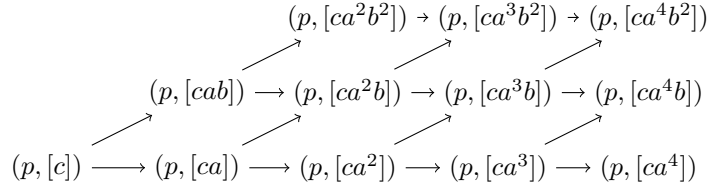
\begin{figure}[h]
    \begin{center}
      \begin{tikzpicture}
        \node (00) at (0,0) {$(p,[c])$};
        \node (10) at (2,0) {$(p,[ca])$};
        \node (20) at (4,0) {$(p,[ca^2])$};
        \node (30) at (6,0) {$(p,[ca^3])$};
        \node (40) at (8,0) {$(p,[ca^4])$};

        \node (11) at (2,1) {$(p,[cab])$};
        \node (21) at (4,1) {$(p,[ca^2b])$};
        \node (31) at (6,1) {$(p,[ca^3b])$};
        \node (41) at (8,1) {$(p,[ca^4b])$};

        \node (22) at (4,2) {$(p,[ca^2b^2])$};
        \node (32) at (6,2) {$(p,[ca^3b^2])$};
        \node (42) at (8,2) {$(p,[ca^4b^2])$};
        
        \graph{ (00) -> (10),
                (10) -> (20),
                (20) -> (30),
                (30) -> (40),
                (11) -> (21),
                (21) -> (31),
                (31) -> (41),
                (22) -> (32),
                (32) -> (42),
                (00) -> (11),
                (11) -> (22),
                (10) -> (21),
                (21) -> (32),
                (20) -> (31),
                (31) -> (42),
                (30) -> (41)};
      \end{tikzpicture}
    \end{center}
    \caption{Configuration graph of the tPDS from Example \ref{ex:tPDS}. }
    \label{fig:ex:tPDS}
  \end{figure}
\end{example}

\begin{example}
  Consider the dependence alphabet $\structD=(A,D)$ with alphabet
  $A=\{a,b,c,d\}$ and dependence relation $D=\{a,b\}^2\cup\{c,d\}^2$.
  Then $(A,I)$ is the
  graph $C_4$. Hence, by Zetzsche's result, there is a pushdown system
  $\autP$ whose reachability relation (in trace semantics) is
  undecidable. In this example, we provide such systems $\autP_1$ and
  $\autP_2$ that satisfy \ref{def:tPDS1} and \ref{def:diamond1'},
  respectively (but not both these conditions).

  To obtain the undecidability, we use Minsky's two-counter
  machines. Here, a two-counter machine is a mapping
  \[
    \autM\colon[n]\to\bigl(\{\INC_1,\INC_2\}\times\{0,\dots,n\}\bigr)
    \cup \bigl(\{\DEC_1,\DEC_2\}\times\{0,1,\dots,n\}^2\bigr)
  \]
  for some $n\in\bbN$. A configuration is a triple $(p,a,b)$ with
  $p\in\{0,\dots,n\}$ and $a,b\in\bbN$ and defines, for every
  configuration $(p,a,b)$ with $p\neq0$ a successor configuration
  $(q,c,d)$ as follows:
  \begin{itemize}
  \item If $\autM(p)=(\INC_1,r)$, then $(q,c,d)=(r,a+1,b)$.
  \item If $\autM(p)=(\INC_2,r)$, then $(q,c,d)=(r,a,b+1)$.
  \item If $\autM(p)=(\DEC_1,r,s)$ and $a>0$, then $(q,c,d)=(r,a-1,b)$.\\
    If $\autM(p)=(\DEC_1,r,s)$ and $a=0$, then $(q,c,d)=(s,a,b)$.
  \item If $\autM(p)=(\DEC_2,r,s)$ and $b>0$, then $(q,c,d)=(r,a,b-1)$.\\
    If $\autM(p)=(\DEC_2,r,s)$ and $b=0$, then $(q,c,d)=(s,a,b)$.
  \end{itemize}
  A natural number $k\in\bbN$ is accepted by $\autM$ if the two-counter
  machine $\autM$ can reach the configuration $(0,0,0)$ when started in
  configuration $(1,k,0)$. By \cite{Min61}, there exists a two-counter
  machine $\autM$ such that the set of accepted natural numbers is
  undecidable. In the following, fix such a two-counter machine $\autM$.
  \begin{enumerate}
  \item States of the PDS $\autP_1$ are ``line numbers'' of $\autM$, i.e.,
    $Q=\{0,1,\dots,n\}$. The configuration $(p,k,\ell)$ of the
    two-counter machine $\autM$ will be encoded by the configuration
    $(p,[a^kb\,c^\ell d])$ of the pushdown system $\autP_1$ with trace
    semantics. We will use that this configuration equals
    $(p,[c^\ell d\,a^kb])$. Therefore, let the set $\Delta$
    of transitions be the least set satisfying the following for all
    $p\in[n]$:
    \begin{itemize}
    \item If $\autM(p)=(\INC_1,r)$, then $(p,a,aa,r),(p,b,ab,r)\in\Delta$.
    \item If $\autM(p)=(\INC_2,r)$, then $(p,c,cc,r),(p,d,cd,r)\in\Delta$.
    \item If $\autM(p)=(\DEC_1,r,s)$, then $(p,a,\varepsilon,r),(p,b,b,s)\in\Delta$.
    \item If $\autM(p)=(\DEC_2,r,s)$, then $(p,c,\varepsilon,r),(p,d,d,s)\in\Delta$.
    \end{itemize}
    Let $\autP_1=(Q,\Delta)$. Since $D(a)=D(b)=\{a,b\}$ and
    $D(c)=D(d)=\{c,d\}$, this pushdown system satisfies
    \ref{def:tPDS1}. We consider its configurations that are reachable
    from the configuration $(1,[a^k b\,d])$ for some $k\in\bbN$. They
    all have the form $(p,[a^\ell b\,c^m d])$ for some
    $\ell,m\in\bbN$. Furthermore,
    $(1,[a^k b\,d])\vdash^*(p,[a^\ell b\,c^m d])$ iff the two-counter
    machine $\autM$ can reach the configuration $(p,\ell,m)$ from
    $(1,k,0)$. It follows that $(1,[a^k b\,d])\vdash^*(0,[b\,d])$ iff
    $k$ is accepted by $\autM$. Since this is undecidable, the
    reachability relation of the PDS with trace semantics $\autP_1$ is
    undecidable as well.
  \item We now construct another pushdown system $\autP_2$ with
    undecidable reachability problem. As in the above construction,
    all ``line numbers'' $p\in\{0,1,\dots,n\}$ are states and the
    two-counter machine's configuration $(p,k,\ell)$ is encoded by the
    configuration $(p,[a^kb\,c^\ell d])$ of the pushdown system with
    trace semantics. But depending on $\autM(p)$ for $p\in[n]$, we will
    add additional states (see below). Then the set $\Delta$ of
    transitions is defined as follows:
    \begin{itemize}
    \item If $\autM(p)=(\INC_1,r)$, then
      $(p,a,ac,p'),(p,b,bc,p'), (p',c,a,r)\in\Delta$.
    \item If $\autM(p)=(\INC_2,r)$, then
      $(p,c,ca,p'),(p,d,da,p'),(p',a,c,r)\in\Delta$.
    \item If $\autM(p)=(\DEC_1,r,s)$, then we have the following
      transitions (where the new states are not named):
      
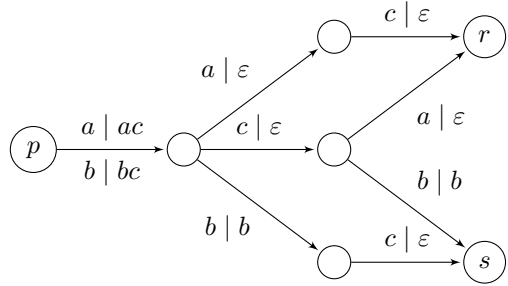
\begin{center}
        \begin{tikzpicture}[nfa]
          \node[state] (p) at (0,0) {$p$};
          \node[state] (p') at (2,0) {};
          \node[state] (q1) at (4,1.5) {};
          \node[state] (q2) at (4,0) {};
          \node[state] (q3) at (4,-1.5) {};
          \node[state] (r) at (6,1.5) {$r$};
          \node[state] (s) at (6,-1.5) {$s$};

          \draw (p) edge node {$a\mid ac$} node[below] {$b\mid bc$} (p')
                (p') edge node[above left] {$a\mid \varepsilon$} (q1)
                (p') edge node {$c\mid \varepsilon$} (q2)
                (p') edge node[below left] {$b\mid b$} (q3)
                (q1) edge node {$c\mid \varepsilon$} (r)
                (q2) edge node[below right] {$a\mid \varepsilon$} (r)
                (q2) edge node[above right] {$b\mid b$} (s)
                (q3) edge node {$c\mid \varepsilon$} (s);
        \end{tikzpicture}
      \end{center}
    \item If $\autM(p)=(\DEC_2,r,s)$, then we have the following
      transitions (where the new states are not named):
      
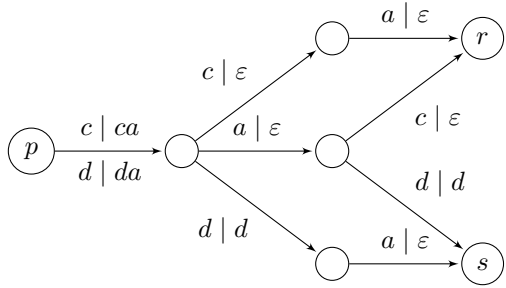
\begin{center}
        \begin{tikzpicture}[nfa]
          \node[state] (p) at (0,0) {$p$};
          \node[state] (p') at (2,0) {};
          \node[state] (q1) at (4,1.5) {};
          \node[state] (q2) at (4,0) {};
          \node[state] (q3) at (4,-1.5) {};
          \node[state] (r) at (6,1.5) {$r$};
          \node[state] (s) at (6,-1.5) {$s$};

          \draw (p) edge node {$c\mid ca$} node[below] {$d\mid da$} (p')
                (p') edge node[above left] {$c\mid \varepsilon$} (q1)
                (p') edge node {$a\mid \varepsilon$} (q2)
                (p') edge node[below left] {$d\mid d$} (q3)
                (q1) edge node {$a\mid \varepsilon$} (r)
                (q2) edge node[below right] {$c\mid \varepsilon$} (r)
                (q2) edge node[above right] {$d\mid d$} (s)
                (q3) edge node {$a\mid \varepsilon$} (s);
        \end{tikzpicture}
      \end{center}
    \end{itemize}
    Let $\autP_2$ denote the result of this construction. Then,
    clearly, $\autP_2$ is a pushdown system. Note that any transition
    starting in a ``line number'' $p$ writes a word $w$ with
    $D(w)=A$. It follows that the PDS $\autP_2$ satisfies
    \ref{def:diamond1'}. We consider its configurations of the form
    $(p,t)$ with $p\in\{0,1,\dots,n\}$ that are reachable from the
    configuration $(1,[a^k b\,d])$ for some $k\in\bbN$. They all have
    the form $(p,[a^\ell b\,c^m d])$ for some
    $\ell,m\in\bbN$. Furthermore,
    $(1,[a^k b\,d])\vdash^*(p,[a^\ell b\,c^m d])$ iff the two-counter
    machine $\autM$ can reach the configuration $(p,\ell,m)$ from
    $(1,k,0)$. It follows that $(1,[a^k b\,d])\vdash^*(0,[b\,d])$ iff
    $k$ is accepted by $\autM$. Since this is undecidable, the
    reachability relation of the PDS with trace semantics $\autP_2$ is
    undecidable as well.\qed
  \end{enumerate}
\end{example}

The example above indicates that the conditions \ref{def:tPDS1} and
\ref{def:diamond1'} are necessary to get the decidability of the
reachability relation.

Next, we simplify the definition of trace-pushdown systems showing
that, in \ref{def:diamond1'}, it suffices to require $a\parallel b$
instead of $av\parallel bw$.

\begin{lemma}\label{lem:(P2)-suffices}
  Let $\structD=(A,D)$ be a dependence alphabet and $\autP=(Q,\Delta)$
  a pushdown system satisfying \ref{def:tPDS1}. Then
  \ref{def:diamond1'} holds if, and only if,
  \begin{enumerate}[label=(P2)]
  \item for each $(p,a,v,q),(q,b,w,r)\in\Delta$ with $a\parallel b$,
    there is a state $q'\in Q$ with
    $(q,b,w,q'),(q',a,v,r)\in\Delta$.\label{def:diamond1}
  \end{enumerate}
\end{lemma}

\begin{proof}
  First, suppose \ref{def:diamond1} holds. To show
  \ref{def:diamond1'}, let $(p,a,v,q),(q,b,w,r)\in\Delta$ with
  $av\parallel bw$. Then, in particular, $a\parallel b$. Hence, by
  \ref{def:diamond1}, there is a state $q'$ as required by
  \ref{def:diamond1'}.

  Conversely, suppose \ref{def:diamond1'} holds and let
  $(p,a,v,q),(q,b,w,r)\in\Delta$ with $a\parallel b$. We want to show
  $av\parallel bw$. Towards a contradiction, let $(c,d)\in D$ where
  $c$ is some letter from $av$ and $d$ some letter from $bw$. Then we
  have $c\in D(d)\subseteq D(bw)=D(b)\cup D(w)=D(b)$ since, by
  \ref{def:tPDS1}, $D(w)\subseteq D(b)$. Consequently, $(c,b)\in D$
  and therefore $b\in D(c)\subseteq D(av)=D(a)\cup D(v)=D(a)$ since,
  by \ref{def:tPDS1}, $D(v)\subseteq D(a)$. Consequently,
  $(a,b)\in D$, contradicting our assumption $a\parallel b$. Thus,
  indeed, $av\parallel bw$. Since $\autP$ also satisfies
  \ref{def:diamond1'}, we obtain some state $q'$ with
  $(p,b,w,q'),(q',a,v,r)\in\Delta$ as required. Hence, indeed, $\autP$
  satisfies \ref{def:diamond1}.
\end{proof}

Let $\autP=(Q,\Delta)$ be a trace-pushdown system and let
$C,D\subseteq\Conf{\autP}$ be two sets of configurations. We write
$C\vdash^*D$ if there are $c\in C$ and $d\in D$ with $c\vdash^*d$,
i.e., if some configuration from $D$ is reachable from some
configuration from $C$. If $C=\{c\}$ ($D=\{d\}$, resp.) is a
singleton, we also write $c\vdash^*D$ ($C\vdash^*d$, resp.). We also
use similar notations for the one-step relation $\vdash$.

We consider the following decision problem: given a trace-pushdown
system and two sets of configurations $C$ and $D$, does
$C\vdash^* D$ hold?  To solve this problem, it is instructive to
first look at its solution for pushdown systems.

So suppose $\autP=(Q,\Delta)$ is a pushdown system and $C,D$ two sets
of configurations. In order to decide whether $C\vdash^* D$, it
suffices to be able to decide whether, for two states $p$ and $q$ and
two languages $K$ and $L$, we have
\begin{equation}
  \label{eq:reachability}
  \{p\}\times K\vdash^*\{q\}\times L\,.
\end{equation}
Let $R=\{(u,v)\mid (p,u)\vdash^*(q,v)\}$. Then \eqref{eq:reachability}
holds iff $K^R\cap L\neq\emptyset$ (which is equivalent to
$K\cap {}^RL\neq\emptyset$).  Caucal \cite{Cau88} gave an algorithm
that, from the transition relation $\Delta$ and two states $p$ and $q$
constructs finitely many regular languages $U_i,V_i,W_i$ such that
$R=\bigcup_{1\le i\le n}(U_i\times V_i)\cdot\{(w,w)\mid w\in W_i\}$,
i.e., he proved that $R$ is effectively prefix-recognizable. From
these regular languages, one can construct a transducer $\autT$
accepting the relation $R$. If $K$ is regular, then an NFA for $K^R$
can be computed from an NFA for $K$ and the transducer $\autT$, i.e.,
$K^R$ is effectively regular. Hence, if also $L$ is regular, one can
decide whether $K^R\cap L\neq\emptyset$, i.e., whether
\eqref{eq:reachability} holds.  Alternatively, if $K$ and $L$ are
regular, then $K\cap {}^RL$ is effectively regular such that one can
decide the emptiness of this language.

Now, suppose $\autP=(Q,\Delta)$ is a trace-pushdown system. Then, as
above, we have to decide whether
\[
  \{p\}\times\calK\vdash^* \{q\}\times\calL
\]
for states $p$ and $q$ and trace languages $\calK$ and $\calL$. Let
$\calR=\{([u],[v])\mid (p,[u])\vdash^*(q,[v])\}$. One can, and we
actually will, show that this relation is a rational trace relation
(but not prefix-recognizable as in the above situation). The problem
occurs in the next step: there are two generalizations (rationality
and recognizability) of regularity to trace languages. But neither of
them serves our purpose since, by Lemma~\ref{lem:nonPreservation},
\begin{itemize}
\item there is a rational trace language $\calL$ and a rational
  trace relation $\calR$ such that ${}^\calR\calL$ is not rational
  and
\item there is a recognizable trace language $\calK$ and a rational
  trace relation $\calR$ such that $\calK^\calR$ is not recognizable.
\end{itemize}

The solution will be to prove that $\calR$ is not just rational, but
even ``lc-rational''. Since this is a new notion, the following
section will introduce and study these lc-rational
relations. Afterwards, we will show that the relation $\calR$ is
indeed lc-rational which will allow to complete the above program.

\section{LC-Rational Trace Relations}
\label{sec:lc-rational}
Above, we explained that we lack a class of rational trace relations
$\calR$ whose left- and right-application preserves rationality and/or
recognizablity. It turns out that the more fundamental property is the
composition of relations, i.e., this section aims at a class
$\bbC_{\traceMon}$ of rational trace relations that is closed under
composition. Recall that a trace relation $\calR$ is rational if, and
only if, there exists a rational word relation $R$ with
$\calR=[R]$. We will therefore first define a class $\bbC_{A^*}$
of rational word relations $R$ that is closed under composition and
satisfies
\begin{equation}
  [R_1\circ R_2]=[R_1]\circ [R_2]
  \label{eq:circ-and-eta}
\end{equation}
for any relations $R_1$ and $R_2$ in $\bbC_{A^*}$ (setting
$\bbC_{\traceMon}=\{[R]\mid R\in\mathbb C_{A^*}\}$ will then
ensure that $\bbC_{\traceMon}$ is closed under composition,
cf.~\cref{def:lcRatTraceRel}).

But first, we show that Eq.~\eqref{eq:circ-and-eta} does not hold for
arbitrary rational word relations; the example also demonstrates that
the composition of two rational trace relations need not be rational.

\begin{example}\label{ex:noComposition}
  Suppose there are $a,b,c,d\in A$ with $(a,b)\in D$ and
  $c\parallel d$. Consider the rational word relations
  \begin{align*}
    R_1&=\{(abab,cdcd)\}\cdot\bigl\{(ab,cd)\}^*=\bigl\{\bigl((ab)^n,(cd)^n\bigr)\bigm| n\ge2\bigr\}\text{ and}\\
    R_2&=\{(c,a)\}^*\,\{(d,b)\}^*=\{(c^m d^n,a^m b^n)\mid m,n\ge0\}\,.
  \end{align*}

  Since $c\parallel d$, we get $(cd)^n\sim c^n d^n$. Consequently
  \begin{align*}
    \calR_1:=[R_1] &=\{([ab]^n,[c^n d^n])\mid n\ge2\}\text{ and}\\
    \calR_2:=[R_2] &=\{([c^m d^n],[a^m b^n])\mid m,n\ge0\}\,.
  \end{align*}
  Hence $\calR:=\calR_1\circ\calR_2=\{([ab]^n,[a^n b^n])\mid n\ge2\}$.
  Thus, we have
  $[R_1\circ R_2]=[\emptyset]=\emptyset\neq \calR=
  \calR_1\circ\calR_2=[R_1]\circ[R_2]$, i.e.,
  Eq.~\eqref{eq:circ-and-eta} does not hold for all rational word
  relations.

  To show that the trace relation $\calR=\calR_1\circ\calR_2$ is not
  rational, let $R\subseteq A^*\times A^*$ be a word relation with
  $[R]=\calR$. Since $(a,b)\in D$, any equivalence class $[v]$ is a
  singleton for $v\in\{a,b\}^*$. Hence $[R]=\calR$ implies
  $R=\{((ab)^n,a^n b^n)\mid n\ge2\}$.  Consider the homomorphism
  $\eta\colon A^*\times A^*\to A^*\colon (u,v)\mapsto v$. Then the
  language $\eta(R)$ equals $\{a^n b^n\mid n\ge2\}$ which is not
  regular (i.e., not rational). Since the class of rational sets is
  closed under homomorphic images, it follows that $R$ and therefore
  the composition $\calR=\calR_1\circ\calR_2$ of two rational trace
  relations is not necessarily rational.
\end{example}

\subsection{LC-Rational Word Relations}

It is well-known (cf.\ \cite{Ber79}) that left- and right-application
of rational word relations to regular languages yield regular
languages, that the composition of rational word relations is
rational, and that the inverse of a rational word relation is
rational. The following theorem summarizes the algorithmic aspects of
these facts.

\begin{theorem}\label{thm:ratWordRel}{\ }
  \begin{enumerate}[label=(R\arabic*)]
  \item From an NFA $\autA$ and a transducer $\autT$, one can compute
    in polynomial time NFAs $\autA_1$ and $\autA_2$ such that
    $L(\autA_1)=L(\autA)^{R(\autT)}$ and
    $L(\autA_2)={}^{R(\autT)}L(\autA)$ (cf.~\cite[Cor.~III.4.2]{Ber79}).
  \item From transducers $\autT_i$ for $1\le i\le n$, one can compute
    a transducer $\autT$ with
    \begin{align*}
      R(\autT)&=R(\autT_1)\circ R(\autT_2)\circ\cdots\circ R(\autT_n)\\
              &=\left\{(u,v)\in A^*\times A^*\,\middle|\,
              \begin{matrix}
                \exists u_i\in A^*\colon u=u_0,u_n=v\\
                (u_{i-1},u_i)\in R(\autT_i)\text{ for all }1\le i\le n
              \end{matrix}
              \right\}\,.
    \end{align*}
    This computation can be carried out in time $O(t^n)$ where
    $t=\sum_{1\le i\le n}\|\autT_i\|$ is the total size of the
    transducers $\autT_i$ (cf.~\cite[Thm.~III.4.4]{Ber79}).
  \item From a transducer $\autT$, one can compute in polynomial time
    a transducer $\autT'$ such that
    $R(\autT')=\{(v,u)\mid (u,v)\in R(\autT)\}$.
  \end{enumerate}
\end{theorem}

\begin{proof}
  We only sketch the construction of the automata and transducers
  claimed to exist.
  \begin{enumerate}[label=(R\arabic*)]
  \item Here, we construct an $\varepsilon$-NFA $\autA_1$ (the
    construction of the $\varepsilon$-NFA $\autA_2$ is analogous). The
    idea is to run the NFA $\autA$ and the transducer $\autT$ in
    parallel; the NFA checks whether the input is accepted and the
    transducer, at the same time, produces its output. So let
    $\autA=(Q,A,I,T,F)$ and $\autT=(Q',A,I',T',F')$. Then
    $\autA_1=(Q\times Q',A,I\times I',T_1,F\times F')$ with
    $((p,p'),b,(q,q'))\in T_1$ iff
    \begin{itemize}
    \item there exists $a\in A$ with $(p,a,q)\in T$ and
      $(p',(a,b),q')\in T'$ or
    \item $p=q$ and $(p',(\varepsilon,b),q')\in T'$.
    \end{itemize}
    Hence we have $(p,p')\xrightarrow{\ v\ }(q,q')$ in
    $\autA_1$ iff there is a word $u\in A^*$ with
    $p\xrightarrow{\ u\ }q$ in $\autA$ and
    $p'\xrightarrow{(u,v)}q'$ in $\autT$.
  \item Let $\autT_i=(Q_i,A,I_i,T_i,F_i)$ for $1\le i\le n$ be
    transducers. We construct a transducer $\autT=(Q,A,I,T,F)$ with
    $Q=\prod_{1\le i\le n}Q_i$, $I=\prod_{1\le i\le n}I_i$, and
    $F=\prod_{1\le i\le n}F_i$. To define the set of transitions, let
    $\overline{p}=(p_i)_{1\le i\le n}$ and
    $\overline{q}=(q_i)_{1\le i\le n}$ be two states from $Q$ and let
    $a,b\in A\cup\{\varepsilon\}$. Then
    $(\overline{p},(a,b),\overline{q})\in T$ if, and only if, one of
    the following hold:
    \begin{itemize}
    \item $a=b=\varepsilon$ and there exist $1\le i<n$ and $c\in A$
      with $(p_i,(\varepsilon,c),q_i)\in T_i$ and
      $(p_{i+1},(c,\varepsilon),q_{i+1})\in T_{i+1}$ and $p_j=q_j$ for
      $j\notin\{i,i+1\}$
    \item $a=b=\varepsilon$ and there exists $1\le i\le n$ with
      $(p_i,(\varepsilon,\varepsilon),q_i)\in T_i$ and $p_j=q_j$ for
      $j\neq i$
    \item $a=\varepsilon$, $(p_n,(\varepsilon,b),q_n)\in T_n$ and
      $p_j=q_j$ for $j<n$
    \item $b=\varepsilon$, $(p_1,(a,\varepsilon),q_1)\in T_1$ and
      $p_j=q_j$ for $j>1$
    \end{itemize}
    By induction on the length of the words $u$ and $v$, one then gets
    $\overline{p}\xrightarrow{(u,v)}\overline{q}$ in $\autT$ if, and
    only if, there are words $u_i$ with $u=u_0$,
    $p_i\xrightarrow{(u_{i-1},u_i)}q_i$ in $\autT_i$ for all
    $1\le i\le n$, and $u_n=v$ implying that $\autT$ has the desired
    semantics. The size of $\autT$ can be estimated by
    \[
      \|\autT\|= \prod_{1\le i\le n} |Q_i|+|A|
      \le \prod_{1\le i\le n} (|Q_i|+|A|)
      = \prod_{1\le i\le n}\|\autT_i\|
      \le t^n
    \]
    and it is easily seen that $\autT$ can be constructed in time
    $O(t^n)$.
  \item Here, we simple switch input and output: a transition
    $(p,(a,b),q)$ is replaced by the transition $(p,(b,a),q)$ which
    can be done in polynomial time.\qedhere
  \end{enumerate}
\end{proof}

Now suppose $\calR\subseteq\traceMon^2$ is a rational trace relation,
i.e., $\calR=[R(\autT)]$ for some transducer $\autT$. Then, by (R3)
above, also $\inv{\calR}=[\inv{R(\autT)}]$ is an efficiently rational
trace relation. Thus, we have an analogue of (R3) in the trace monoid,
similar analogues of (R1) and (R2) fail. Regarding (R1), there are two
possible formulations in the trace monoid, but
\cref{lem:nonPreservation} will demonstrate that the application of
rational relations does neither preserve the rationality nor the
recognizability of a trace language. Regarding (R2),
\cref{ex:noComposition} above demonstrates that the class of
rational trace relations is not closed under composition.

Note that Eq.~\eqref{eq:circ-and-eta} fails in \cref{ex:noComposition}
since there, $R_1\circ R_2=\emptyset$ and
$\calR_1\circ\calR_2\neq\emptyset$. The reason is that there are words
$t,u,u',v'\in A^*$ such that $(t,u)\in R_1$, $(u',v')\in R_2$, and
$u\sim u'$ distinct such that $([t],[v'])\in\calR_1\circ\calR_2$, but
$(t,v')\notin R_1\circ R_2$.  The following definition circumvents
this problem.

\begin{definition}
  A relation $R\subseteq A^*\times A^*$ is \emph{left-closed}
  if $\mathord{\sim}\circ R\subseteq R\circ\mathord{\sim}$, i.e., 
  \[
  \bigl(\exists u'\in A^*\colon u\sim u'\relation{R} v'\bigr)\ \Longrightarrow\
  \bigl(\exists v\in A^*\colon u\relation{R}v\sim v'\bigr)
  \]
  holds for all $u,v'\in A^*$.  The relation $R$ is \emph{lc-rational}
  if it is left-closed and rational. We call a transducer $\autT$
  \emph{left-closed} or an \emph{lc-transducer} if the relation
  $R(\autT)$ is left-closed.
\end{definition}

Note that the definition of a left-closed transducer is based on the
relation accepted by the transducer. 

A very simple example for an lc-rational relation is the identity relation
$\Id_{A^*}=\{(u,u)\mid u\in A^*\}$: if
$u\sim u'\relation{\Id_{A^*}}v'$, then $u'=v'$. Setting $v=u$,
we obtain $u\relation{\Id_{A^*}}v=u\sim u'=v'$. Since this relation
is clearly rational, it is indeed lc-rational. Other examples are
$A^*\times\{\varepsilon\}$ and $\{\varepsilon\}\times A^*$.

\begin{example}\label{ex:superword}
  A word $u\in A^*$ is a \emph{subword} of $v\in A^*$ if
  $u=u_1 u_2\cdots u_n$ and $v=u_1 v_1 u_2 v_2\cdots u_n v_n$ for some
  $n\in\monoN$ and $u_1,u_2,\dots,u_1,v_1,v_2,\dots,v_n\in A^*$. In this
  case, we write $u\preceq v$. The subword-relation is rational since
  $\mathord{\preceq}=\{(a,a),(\varepsilon,a)\mid a\in A\}^*$.

  Suppose $a\parallel c$, $a\nparallel b$, and $b\nparallel c$.  Then
  $ca\sim ac\preceq abc$, but there is no superword of $ca$ that is
  equivalent to~$abc$. Hence the subword-relation is not left-closed.
  
  We now show that the inverse relation $\succeq$ (the superword
  relation) is left-closed. Suppose $u=x\,ab\,y$ with
  $a\parallel b$ and $x\,ba\,y=u'\succeq v'$. If $v'=\varepsilon$, we get
  $u\succeq\varepsilon\sim v'$. So suppose $v'\neq\varepsilon$. Then there
  are $n\in\monoN$, $u_1\in A^*$,
  $u_2,\dots,u_n,v_1,v_2,\dots,v_n\in A^+$, and $u_{n+1}\in A^*$ such
  that $v_1v_2\cdots v_n=v'$ and
  $u_1v_1u_2v_2\cdots u_nv_nu_{n+1}=u'=x\,ba\,y$. This gives two
  decompositions of the word $u'$ into the blocks $u_i$ and $v_i$ on
  the one hand, and into the factors $x$, $ba$, and $y$ on the other
  hand. The factor $ba$ from the second factorization can be covered
  by some factor $u_i$, by some factor $v_i$, or it belongs to some
  consecutive factors $u_i$ and $v_i$ or $v_i$ and $u_{i+1}$ of the
  former factorization---in any of these cases, we construct a word
  $v$ with $u\succeq v\sim v'$:
  \begin{itemize}
    \item Suppose $x=u_1v_1\cdots u_{i-1}v_{i-1}x'$, $u_i=x'bay'$, and
    $y=y'v_i u_{i+1}v_{i+1}\cdots u_{n+1}$. Then 
    \begin{align*}
      u &= u_1v_1u_2\cdots u_{i-1}v_{i-1}\, x'ab y'\,v_i u_{i+1}v_{i+1}\cdot u_{n+1}\\
      &\succeq v_1v_2\cdots v_n=v'\,.
    \end{align*}
    Thus, setting $v:=v'$ yields $u\succeq v\sim v'$.
    \item Suppose $x=u_1v_1\cdots u_{i-1}v_{i-1}u_i x'$, $v_i=x'bay'$, and
    $y=y' u_{i+1}v_{i+1}\cdots u_{n+1}$. Then 
    \begin{align*}
      u &=u_1v_1\cdots u_{i-1}v_{i-1}u_i\,x'aby'\,u_{i+1}v_{i+1}\cdots v_n u_{n+1}\\
      &\succeq v_1\cdots v_{i-1}\,x'aby'\,v_{i+1}\cdots v_n \\
      &\sim v_1\cdots v_{i-1}\,v_i\,v_{i+1}\cdots v_n = v'\,.
    \end{align*}
    Thus, setting $v:=v_1\cdots v_{i-1}\,x'aby'\,v_{i+1}\cdots v_n$ yields
    $u\succeq v\sim v'$.
    \item Suppose $x=u_1v_1\cdots u_{i-1}v_{i-1}x'$, $u_i=x'b$,
    $v_i=ay'$, and $y=y'u_{i+1}v_{i+1}\cdots u_{n+1}$. Then
    \begin{align*}
      u&=u_1v_1\cdots u_{i-1}v_{i-1}\,x'b\,
      ay'\,u_{i+1}v_{i+1}\cdots u_{n+1}\\
      &\succeq v_1\cdots v_{i-1}\,ay'\,v_{i+1}\cdots v_n=v'\,.  
    \end{align*}
    Thus, setting $v:=v'$ yields $u\succeq v\sim v'$.
    \item Finally, suppose $x=u_1v_1\cdots u_ix'$, $v_i=x'b$,
    $u_{i+1}=ay'$, and $y=y' v_{i+1}u_{i+2}v_{i+2}\cdots u_{n+1}$. Then
    \begin{align*}
      u&=u_1v_1\cdots u_i\,x'b\,ay'\,v_{i+1}u_{i+2}v_{i+2}\cdots u_{n+1}\\
      &\succeq v_1\cdots v_{i-1}\,x'b\,v_{i+1}\cdots v_{n+1} = v\,.
    \end{align*}
    Thus, setting $v:=v'$ yields $u\succeq v\sim v'$.
  \end{itemize}
  
  Since $\sim$ is the
  least equivalence relation identifying $xaby$ with $xbay$ for
  $a\parallel b$, this proves that the superword-relation $\succeq$ is
  left-closed.
\end{example}

We next show that the class of lc-rational word relations has the
desired properties: it is closed under composition and the
homomorphism $[.]$ commutes with composition.

\begin{proposition}\label{prop:leftClosednessAndComposition}
  Let $R_1,R_2\subseteq A^*\times A^*$.
  \begin{enumerate}[label=(\roman*)]
  \item If $R_2$ is left-closed, then Eq.~\eqref{eq:circ-and-eta}
    holds, i.e., $[R_1\circ R_2]=[R_1]\circ[R_2]$.
  \item If $R_1$ and $R_2$ are lc-rational, then $R_1\circ R_2$ is
    lc-rational. More precisely, from left-closed transducers
    $\autT_i$ for $1\le i\le n$, one can compute a left-closed
    transducer $\autT$ with
    \[
      [R(\autT)]=[R(\autT_1)]\circ [R(\autT_2)]\circ\cdots[R(\autT_n)]
    \]
    in time $O(t^n)$ where $t$ is the total size of the transducers
    $\autT_i$.
  \end{enumerate}
\end{proposition}

\begin{proof}
  To demonstrate the first claim, let $R_2$ be lc-rational.  For the
  inclusion $[R_1\circ R_2] \subseteq [R_1]\circ[R_2]$, let
  $(u,w)\in R_1\circ R_2$. Then there exists $v\in A^*$ with
  $u\relation{R_1}v\relation{R_2}w$ and therefore
  $[u] \relation{[R_1]} [v] \relation{[R_2]} [w]$ implying
  $([u],[w])\in [R_1]\circ [R_2]$.
  
  For the converse inclusion, let
  $(x,z)\in[R_1]\circ[R_2]$. There is some trace $y$ with
  $x\relation{[R_1]}y\relation{[R_2]}z$. Hence there are words
  $u,v,v',w$ with
  \begin{itemize}
    \item $x=[u]$, $y=[v]$, and $(u,v)\in R_1$ and
    \item $y=[v']$, $z=[w]$, and $(v',w)\in R_2$.
  \end{itemize}
  Hence we have $u\relation{R_1} v\sim v'\relation{R_2}w$.
  Since $R_2$ is left-closed, there exists a word $w'\in A^*$
  such that $u\relation{R_1} v\relation{R_2}w'\sim w$.  Hence we have
  $(x,z)=([u],[w])=([u],[w'])\in [R_1\circ R_2]$.
  This finishes the verification of the first claim.
  
  Now, assume both relations $R_1$ and $R_2$ to be left-closed such
  that $\mathord{\sim}\circ R_i\subseteq R_i\circ\mathord{\sim}$ holds
  for all $i\in[2]$. Consequently, we get
  $\mathord{\sim}\circ R_1\circ R_2 \subseteq
  R_1\circ\mathord{\sim}\circ R_2 \subseteq R_1\circ
  R_2\circ\mathord{\sim}$.  Hence, indeed, $R_1\circ R_2$ is
  left-closed such that the second claim follows using
  Theorem~\ref{thm:ratWordRel}(R2) and statement (i).
\end{proof}

The following proposition characterizes the lc-rational word relations of
the form $K\times L$ for languages $K,L\subseteq A^*$. This
characterization should also explain the name ``left-closed''. 

\begin{proposition}\label{prop:wordsDirectProducts}
  Let $K,L\subseteq A^*$ be nonempty.
  \begin{enumerate}[label=(\roman*)]
  \item Then $K\times L$ is rational if, and only if, $K$ and $L$ both
    are regular. More precisely, we have the following.
      \begin{itemize}
      \item From a transducer $\autT$, one can compute in polynomial
        time NFAs $\autA_1$ and $\autA_2$ with
        $L(\autA_1)=\{u\mid \exists v\colon(u,v)\in R(\autT)\}$ and
        $L(\autA_2)=\{v\mid \exists u\colon(u,v)\in R(\autT)\}$.
      \item Conversely, from two NFAs $\autA_1$ and $\autA_2$, one can construct
        in polynomial time a transducer $\autT$ with
        $R(\autT)=L(\autA_1)\times L(\autA_2)$.
      \end{itemize}
    \item $K\times L$ is left-closed if, and only if, $K$ is
      closed. More precisely, we have the following.
      \begin{itemize}
      \item If $R\subseteq A^*\times A^*$ is left-closed, then
        $\{u\in A^*\mid\exists v\colon(u,v)\in R\}$ is closed.
      \item If $K$ is closed, then $K\times L$ is left-closed.
      \end{itemize}
  \end{enumerate}
  Consequently, $K\times L$ is lc-rational if, and only if, $K$ and
  $L$ both are regular and $K$ is closed.
\end{proposition}

\begin{proof}
  Fix an NFA $\autA$ accepting $A^*$. Now let $\autT=(Q,A,I,T,F)$ be a
  transducer. Then, by Theorem~\ref{thm:ratWordRel}(R1), one can
  construct in polynomial time NFAs $\autA_1$ and $\autA_2$ with
  $L(\autA_1)={}^{R(\autT)}L(\autA)={}^{R(\autT)}(A^*)=\{u\mid \exists
  v\colon (u,v)\in R(\autT)\}$ and, similarly,
  $L(\autA_2)=L(\autA)^{R(\autT)}=(A^*)^{R(\autT)}=\{v\mid \exists
  u\colon (u,v)\in R(\autT)\}$.

  Regarding the second claim, let $\autA_i=(Q_i,A,I_i,T_i,F_i)$ be
  NFAs. Construct transducers $\autT_i=(Q_i,A,I_i,T_i',F_i)$ by
  setting $(p,(a,b),q)\in T_1'$ iff $(p,a,q)\in T_1$ and
  $b=\varepsilon$ as well as $(p,(a,b),q)\in T_2'$ iff
  $(p,b,q)\in T_2$ and $a=\varepsilon$. Then
  $R(\autT_1)=L(\autA_1)\times\{\varepsilon\}$ and
  $R(\autT_2)=\{\varepsilon\}\times L(\autA_2)$. Furthermore, the
  transducers $\autT_i$ can be constructed in polynomial time. The
  second claim now follows from Theorem~\ref{thm:ratWordRel}(R2) since
  $L(\autA_1)\times L(\autA_2)=R(\autT_1)\circ R(\autT_2)$.

  Suppose $R\subseteq A^*\times A^*$ to be left-closed and let
  $K=\{u\mid\exists v\colon(u,v)\in R\}$. To show that $K$ is closed,
  let $u\sim u'\in K$. By the definition of $K$, there exists
  $v'\in A^*$ with $(u',v')\in R$, i.e., $u\sim
  u'\relation{R}v'$. Since $R$ is assumed to be left-closed, there
  exists $v\in A^*$ with $u\relation{R}v\sim v'$. This implies in
  particular $u\in K$. Hence, $K$ is closed.
  
  Conversely, suppose $K$ to be closed and let
  $u\sim u'\relation{(K\times L)}v'$. Then $u\sim u'\in K$ implying
  $u\in K$ such that (with $v=v'$) we get
  $u\relation{(K\times L)}v\sim v'$, i.e., $K\times L$ is
  left-closed.
\end{proof}

\subsection{LC-Rational Trace Relations}

Recall that a trace relation $\calR\subseteq\traceMon\times\traceMon$ is
rational if there exists a rational word relation
$R\subseteq A^*\times A^*$ with $[R]=\calR$. Similarly, we now lift the
concept of lc-rational relations from words to traces.

\begin{definition}\label{def:lcRatTraceRel}
  A trace relation $\calR\subseteq \traceMon\times \traceMon$ is
  \emph{lc-rational} if there exists some lc-rational word relation
  $R\subseteq A^*\times A^*$ with $\calR=[R]$.
\end{definition}

Simple examples are $\traceMon\times\{[\varepsilon]\}$ and
$\{[\varepsilon]\}\times\traceMon$ since $A^*\times\{\varepsilon\}$ and
$\{\varepsilon\}\times A^*$ are lc-rational word relations.

\begin{example}
  Another, more involved example, is the supertrace-relation
  \cite{Kus20}: $x\in\traceMon$ is a \emph{supertrace} of $y\in\traceMon$ if
  $x=x_1 y_1\cdots x_n y_n x_{n+1}$ and $y=y_1y_2\cdots y_n$ for 
  some $n\in\monoN$ and
  $x_1,x_2,\dots,x_{n+1},y_1,y_2,\dots,y_n\in\traceMon$. In this case,
  we write $x\sqsupseteq y$. It is easily checked that
  $x\sqsupseteq y$ if, and only if, there are words $u$ and $v$ such
  that $x=[u]$, $y=[v]$, and $u\succeq v$, i.e.,
  $\mathord{\sqsupseteq}=[\mathord{\succeq}]$. Since the
  superword-relation $\succeq$ is lc-rational by
  \cref{ex:superword}, we obtain that the supertrace-relation
  $\sqsupseteq$ is lc-rational.
\end{example}

Also the identity relation
$\Id_{\traceMon}=\{(x,x)\mid x\in\traceMon\}$ is lc-rational since
$\Id_{A^*}$ is an lc-rational word relation.

By \cref{ex:noComposition}, the composition of rational trace
relations is, in general, not rational. The following proposition
demonstrates that the composition is rational provided the second
relation is lc-rational (differently, if the first relation is
lc-rational, the composition is not necessarily rational as
\cref{ex:noComposition} demonstrates as well). If both relations are
lc-rational, then the composition is so as well.

\begin{proposition}\label{prop:composition}
  Let $\calR_1,\calR_2\subseteq\traceMon\times\traceMon$ be rational
  trace relations.
  \begin{enumerate}[label=(\roman*)]
  \item If $\calR_2$ is lc-rational, then $\calR_1\circ\calR_2$ is
    rational. More precisely, from a transducer $\autT_1$ and a
    left-closed transducer $\autT_2$, one can compute in polynomial
    time a transducer $\autT$ such that
    $[R(\autT)]=[R(\autT_1)]\circ [R(\autT_2)]$.
  \item If $\calR_1$ and $\calR_2$ both are lc-rational, then
    $\calR_1\circ\calR_2$ is even lc-rational. More precisely, from
    left-closed transducers $\autT_i$ for $1\le i\le n$, one can
    compute a left-closed transducer $\autT$ with
    $[R(\autT)]=[R(\autT_1)]\circ [R(\autT_2)]\circ\cdots \circ
    [R(\autT_n)]$ in time $O(t^n)$ where $t$ is the total size of the
    transducers $\autT_i$.
  \end{enumerate}
\end{proposition}

\begin{proof}
  Let $\autT_1$ be a transducer and $\autT_2$ a left-closed
  transducer.  By \cref{thm:ratWordRel}(R2), we can compute in
  polynomial time a transducer $\autT$ with
  $R(\autT)=R(\autT_1)\circ R(\autT_2)$. Since the relation
  $R(\autT_2)$ is left-closed,
  \cref{prop:leftClosednessAndComposition}(i) implies
  \[
    [R(\autT)]=\bigl[R(\autT_1)\circ R(\autT_2)\bigr]
    =[R(\autT_1)]\circ [R(\autT_2)]\,.
  \]

  Next consider $n$ left-closed transducers $\autT_i$. By
  \cref{thm:ratWordRel}(R2), one can compute a transducer $\autT$ in
  the given time bound that accepts the composition of the relations
  $R(\autT_i)$. By \cref{prop:leftClosednessAndComposition}(ii), this
  composition is lc-rational (implying that $\autT$ is a left-closed
  transducer) and, by \cref{prop:leftClosednessAndComposition}(i),
  satisfies
  \[
    [R(\autT)]=
    \bigl[R(\autT_1)\circ R(\autT_2)\circ\cdots\circ R(\autT_n)\bigr]=
    [R(\autT_1)]\circ [R(\autT_2)]\circ\cdots\circ [R(\autT_n)]\,.
  \]
\end{proof}

Next, we want to characterize the lc-rational relations among the
direct products $\calK\times\calL$ of sets of traces $\calK$ and
$\calL$ (we have done so for word relations in
\cref{prop:wordsDirectProducts}). 

\begin{proposition}\label{prop:directProductTraces}
  Let $\calK,\calL\subseteq\traceMon$ be nonempty.
  \begin{enumerate}[label=(\roman*)]
  \item $\calK\times\calL$ is rational if, and only if, $\calK$ and
    $\calL$ both are rational.
  \item $\calK\times\calL$ is lc-rational if, and only if, $\calK$ is
    recognizable and $\calL$ is rational.
  \end{enumerate}
  More precisely, we have the following.
  \begin{enumerate}[label=(\alph*)]
  \item From a transducer $\autT$, one can compute in polynomial time
    NFAs $\autA_1$ and $\autA_2$ with
    $[L(\autA_1)]=\{[u]\mid\exists v\colon([u],[v])\in[R(\autT)]\}$
    and
    $[L(\autA_2)]=\{[v]\mid\exists
    u\colon([u],[v])\in[R(\autT)]\}$. If $\autT$ is a left-closed
    transducer, then $\autA_1$ is a closed NFA (implying that
    $[L(\autA_1)]$ is recognizable).
  \item Conversely, from two NFAs $\autA_1$ and $\autA_2$, one can
    construct in polynomial time a transducer $\autT$ with
    $[R(\autT)]=[L(\autA_1)]\times[L(\autA_2)]$. If $\autA_1$ is
    closed, then $\autT$ is left-closed.
  \end{enumerate}
\end{proposition}

\begin{proof}
  Let $\autT$ be a transducer and let
  $\calK=\{[u]\mid\exists[v]\colon([u],[v])\in[R(\autT)]\}$. By
  \cref{prop:wordsDirectProducts}(i), we can construct in polynomial
  time an NFA $\autA_1$ that accepts
  $K:=\{u\mid \exists v\colon(u,v)\in R(\autT)\}$.  To show
  $[L(\autA_1)]=\calK$, let $u\in A^*$ be some word. Then we have the
  following
  \begin{align*}
    [u]\in\calK &\iff \exists v\in A^*\colon([u],[v])\in [R(\autT)]\\
                &\iff \exists u',v'\colon u\sim u'\relation{R(\autT)}v'\\
                &\iff\exists u'\colon u\sim u'\in K\\
    &\iff [u]\in[K]=[L(\autA_1)]\,.
  \end{align*}
  Hence, indeed,
  $[L(\autA_1)]=\{[u]\mid \exists
  [v]\colon([u],[v])\in[R(\autT)]\}$. This finishes the proof of 
  claim~(a) regarding $\autA_1$, the NFA $\autA_2$ can be obtained
  symmetrically.

  Now let $\autT$ be a left-closed transducer. Then $R(\autT)$ is a
  left-closed relation. Hence, by
  Prop.~\ref{prop:wordsDirectProducts}(ii), the language
  $\{u\mid \exists v\colon(u,v)\in R(\autT)\}$ is closed. But this
  language equals $L(\autA_1)$, i.e., $\autA_1$ is a closed NFA.
  This finishes the proof of claim (a).
  
  To also prove (b), let $\autA_1$ and $\autA_2$ be NFAs. Then, by
  \cref{prop:wordsDirectProducts}(i), one can construct in polynomial
  time a transducer $\autT$ with
  $R(\autT)=L(\autA_1)\times L(\autA_2)$. By the very definition of
  $[R]$ for a word relation $R$, we get
  \[
    [L(\autA_1)]\times[L(\autA_2)]=
    \bigl[L(\autA_1)\times L(\autA_2)\bigr]=[R(\autT)]\,.
  \]
  Now suppose $\autA_1$ to be closed. To show that the transducer
  $\autT$ is left-closed, we have to prove that the relation
  $R(\autT)$ is left-closed. But this is the case by
  Prop.~\ref{prop:wordsDirectProducts}(ii) since
  $R(\autT)=L(\autA_1)\times L(\autA_2)$ and $L(\autA_1)$ is closed.
\end{proof}

%%%%%%%%%%%%%%%%%%%%%%%%%%%%%%%%%%%%%%%%%%%%%%%%%%%%%%%%%%%%%%%%%%%%%%

Proposition~\ref{prop:composition}(i) ensures that the composition of
a rational and an lc-rational trace relation is rational, again. This
holds in particular if the first relation is the inverse of an
lc-rational relation. We now demonstrate that all rational trace
relations arise in this way (provided there are at least two dependent
letters).

\begin{proposition}\label{prop:ratRelsplit}
  Suppose there are $a,b\in A$ distinct with $(a,b)\in D$. Let
  $\calR\subseteq\traceMon^2$ be a rational trace relation.  There exist
  lc-rational trace relations $\calR_1$ and $\calR_2$ such that
  $\calR=\inv{\calR_1}\circ\calR_2$.
\end{proposition}

\begin{proof}
  There exists a rational relation $R\subseteq A^*\times A^*$ with
  $\calR=[R]$. By Nivat's theorem \cite{Niv68}
  (cf.~\cite[Theorem~III.3.2]{Ber79}), there exist an alphabet $B$,
  homomorphisms $f,g\colon B^*\to A^*$, and a regular language
  $K\subseteq B^*$ such that
  \[
    R=\{(f(u),g(u))\mid u\in K\}\,.
  \]
  Suppose $B=\{c_1,c_2,\dots,c_n\}$. Let $h\colon B^*\to A^*$
  be the homomorphism defined by $h(c_i)=a^i b$.
  
  Now consider the relations
  \[
    R_1=\{(h(u),f(u))\mid u\in K\}\text{ and }
    R_2=\{(h(u),g(u))\mid u\in K\}\,.
  \]
  We first show that these relations are lc-rational (by symmetry, we
  only consider the relation~$R_1$). From Nivat's theorem, we obtain
  that $R_1$ is rational. To show that it is left-closed, let
  $v,v',w'\in A^*$ with $v\sim v'\relation{R_1}w'$.
  
  Since $R_1\subseteq\{a,b\}^*\times A^*$, we obtain
  $v'\in\{a,b\}^*$. Since $(a,b)\in D$, this implies $v=v'$. Hence,
  setting $w:=w'$, we obtain $v\relation{R_1}w\sim w'$. Hence,
  indeed, the relations $R_1$ and $R_2$ are lc-rational.
  
  Next, we show $R=\inv{R_1}\circ R_2$. For the inclusion
  ``$\subseteq$'', let $(v,w)\in R$. Then there exists $u\in K$ with
  $v=f(u)$ and $w=g(u)$. Hence we obtain
  \[
    v=f(u)\relation{\inv{R_1}}h(u)\relation{R_2}g(u)=w
  \]
  and therefore $(v,w)\in \inv{R_1}\circ R_2$. For the converse
  inclusion, suppose $(v,w)\in \inv{R_1}\circ R_2$. Then there exists
  some word $x$ with $v\relation{\inv{R_1}}x\relation{R_2}w$. By the
  definition of the relations $R_1$ and $R_2$, there are words
  $u_1,u_2\in K\subseteq B^*$ such that
  \[
    v=f(u_1)\,, x=h(u_1)\text{ and }x=h(u_2)\,, w=g(u_2)\,.
  \]
  Since the homomorphism $h$ is injective, we get $u_1=u_2$ and therefore
  \[
    (v,w)=(f(u_1),g(u_2))=(f(u_1),g(u_1))\in R\,.
  \]
  Thus, indeed, $R=\inv{R_1}\circ R_2$.
  
  Finally, let $\calR_1=[R_1]$ and $\calR_2=[R_2]$. Note that
  $\inv{R_1}$ is rational and satisfies
  $[\inv{R_1}]=\inv{\calR_1}$. From
  \cref{prop:leftClosednessAndComposition}(i), we
  obtain that
  \[
    [\inv{R_1}\circ R_2]=[\inv{R_1}]\circ[R_2]
    =\inv{\calR_1}\circ\calR_2
  \]
  since $R_2$ is lc-rational. Hence, we obtain
  \[
    \calR =[R]
    =[\inv{R_1}\circ R_2]
%    =[\in``v{R_1}]\circ\calR_2
    =\inv{\calR_1}\circ\calR_2\,.\qedhere
  \]
\end{proof}
%%%%%%%%%%%%%%%%%%%%%%%%%%%%%%%%%%%%%%%%%%%%%%%%%%%%%%%%%%%%%%%%%%%%%%

\subsection{Products of lc-rational trace relations}
\label{sec:products}

The (componentwise) product of two rational trace relations is
rational again: if $\calR_1,\calR_2\subseteq\traceMon^2$ are rational,
then there are transducers $\autT_1$ and $\autT_2$ with
$\calR_i=[R(\autT_i)]$. From these transducers, one can construct a
transducer $\autT$ with $R(\autT)=R(\autT_1)\cdot R(\autT_2)$. It
follows that
$\calR_1\cdot\calR_2=[R(\autT_1)]\cdot [R(\autT_2)]=[R(\autT_1)\cdot
R(\autT_2)]=[R(\autT)]$ is rational.
The following lemma demonstrates that this does not hold for lc-rational
relations.

\begin{lemma}\label{lem:noProduct}
  There exist lc-rational relations $\calR$ and $\calR'$ such that
  $\calR\cdot\calR'$ is not lc-rational.
\end{lemma}

\begin{proof}
  Consider the rational trace relations $\calR_1$ and $\calR_2$ from
  \cref{ex:noComposition}. Note that $\calR_2$ is the product of the
  lc-rational trace relations $\calR=\{([c],[a])\}^*$ and
  $\calR'=\{([d],[b])\}^*$ and recall that $\calR_1\circ\calR_2$ is not
  rational. Hence, by \cref{prop:composition}(i),
  $\calR_2=\calR\cdot \calR'$ cannot be lc-rational.
\end{proof}

We now come to two special cases of relations $\calR_1$ that ensure the
lc-rationality of $\calR_1\cdot\calR_2$:

\begin{lemma}\label{lem:productEpsilon1}
  Let $\calK\subseteq\traceMon$ be recognizable. Then the relation
  \[
    \calR=(\calK\times\{[\varepsilon]\})\cdot\Id_{\traceMon}
    =\{(xy,y)\mid x\in\calK,y\in\traceMon\}
  \]
  is lc-rational.

  More precisely, from a dependence alphabet $\structD=(A,D)$ and a
  closed NFA $\autA=(Q,A,I,T,F)$, one can compute a left-closed
  transducer $\autT$ with
  $[R(\autT)]=([L(\autA)]\times\{[\varepsilon]\})\cdot\Id_{\traceMon}$;
  this computation can be carried out in time polynomial in the size
  $\|\autA\|$ of $\autA$ and the set twin index $\settwinindex(\structD)$ of
  $\structD$.
\end{lemma}

\begin{proof}
  Let $R$ denote the set of pairs of words
  \[
    (u_1v_1u_2v_2\cdots u_nv_n,v_1v_2\cdots v_n)
  \]
  with $n\in\monoN$ and $u_1,u_2,\dots,u_n,v_1,\dots,v_n\in A^*$ such
  that
  \begin{enumerate}[label=(\roman*)]
    \item $u_1u_2\cdots u_n\in L(\autA)$,\label{prf:productEpsilon1:i1}
    \item $v_1v_2\cdots v_i\parallel u_{i+1}$ for all $i\in[n-1]$.
      \label{prf:productEpsilon1:i3}
  \end{enumerate}
  We construct a transducer $\autT$ with $R(\autT)=R$, prove
  $[R]=([L(\autA)]\times\{[\varepsilon]\})\cdot\Id_{\traceMon}$, and
  show that $R$ is left-closed (implying that $\autT$ is a left-closed
  transducer).
  
  States of the transducer $\autT$ are triples $(q,D(B),a)$ of a state
  $q\in Q$, a set of letters $B\subseteq A$, and
  $a\in A\cup\{\varepsilon\}$. A state $(q,D(B),a)$ is initial if
  $q\in I$, $D(B)=\emptyset$ (i.e., $B=\emptyset$), and
  $a=\varepsilon$. It is accepting if $q\in F$ and
  $a\in\varepsilon$. The transducer has two types of transitions (with
  $a\in A$):
  \begin{itemize}
  \item There are transitions from $(p,D(B),\varepsilon)$ via
    $(p,D(B),a)$ to $(q,D(C),\varepsilon)$ labeled $(a,\varepsilon)$
    and $(\varepsilon,a)$, respectively iff $p=q$ and
    $D(C)=D(B)\cup D(a)$
  \item There is a transition from $(p,D(B),\varepsilon)$ to
    $(q,D(C),\varepsilon)$ labeled $(a,\varepsilon)$ iff
    $(p,a,q)\in T$ is a transition of the automaton $\autA$,
    $a\notin D(B)$, and $D(B)=D(C)$.
  \end{itemize}
  Clearly, this transducer can be computed in time polynomial in
  $\|\autA\|\cdot\settwinindex(\structD)$.

  Let $p\in I$, $q\in Q$, $B\subseteq A$, and $v,w\in A^*$. Then the
  transducer $\autT$ has a path labeled $(w,v)$ from
  $(p,\emptyset,\varepsilon)$ to $(q,D(B),\varepsilon)$ iff the
  following hold:
  \begin{itemize}
  \item $D(B)=D(\alphabet{v})$ is the set of letters dependent from
    some letter of $v$.
  \item $w$ results from $v$ by injecting some letters (using
    transitions of the second type $(a,\varepsilon)$) that are
    independent from all letters of $v$ read so far.
  \item The sequence $u$ of injected letters leads from $p$ to $q$ in
    the automaton $\autA$.
  \end{itemize}
  Consequently, $(w,v)$ labels a path from some initial to some
  accepting state iff $(w,v)\in R$. Hence, indeed, $R=R(\autT)$ is
  rational.\medskip
  
  Next, we verify
  $[R]=([L(\autA)]\times\{[\varepsilon]\})\cdot\Id_{\traceMon}$. First,
  suppose $(w,v)=(u_1v_1\cdots u_nv_n,v_1\cdots v_n)\in R$ with the
  properties from above. From \ref{prf:productEpsilon1:i3}, we
  obtain $u_1v_1\cdots u_nv_n \sim u_1u_2\cdots u_n\,v_1v_2\cdots v_n$
  and therefore $([w],[v])=([u_1\cdots u_n]\cdot [v],[v])$ which
  belongs to $\calR$ since $u_1\cdots u_n\in K$ by
  \ref{prf:productEpsilon1:i1}. Thus, $[R]\subseteq\calR$.
  Conversely, let $([uv],[v])\in\calR$, i.e., $u\in K$ and $v\in
  A^*$. With $n=1$, $u_1=u$, and $v_1=v$, we get
  $(uv,v)=(u_1v_1,v_1)\in R$ and therefore $\calR\subseteq[R]$. Thus,
  indeed, $[R]=\calR$.\medskip
  
  It remains to be shown that $R$ is left-closed. So let $n\in\monoN$
  be a number and $u_1,\dots,u_n,v_1,\dots,v_n\in A^*$ words satisfying
  \ref{prf:productEpsilon1:i1}-\ref{prf:productEpsilon1:i3} from
  above and let $w\in A^*$ such that
  \[
    w\sim u_1v_1u_2v_2\cdots u_nv_n\relation{R}v_1\cdots v_n\,.
  \]
  With $u=u_1\cdots u_n$ and $v=v_1\cdots v_n$,
  \ref{prf:productEpsilon1:i3} implies
  \[
    w\sim u_1v_1u_2v_2\cdots u_nv_n \sim
    u_1\cdots u_n\,v_1\cdots v_n = uv\,.
  \]
  Application of Levi's Lemma for traces \cite[p.~74]{DieR95} to the
  equivalence $w\sim uv$ yields $m\in\monoN$ and words
  $u'_1,\dots,u_m'$ and $v'_1,\dots,v'_m$ such that
  \begin{enumerate}[label=(\arabic*)]
    \item $w=u_1'v_1'\,u_2'v_2'\cdots u_m'v_m'$,
    \item $u\sim u_1' u_2'\cdots u_m'=:u'$,
    \item $v\sim v_1' v_2'\cdots v_m'=:v'$, and
    \item $v_1'v_2'\cdots v_i'\parallel u_{i+1}'$ for all $i\in[m-1]$.
  \end{enumerate}
  Note that $u'\sim u\in L(\autA)$ implies $u'\in L(\autA)$ since the
  NFA $\autA$ is closed. Hence we get $w\relation{R}v'\sim
  v$. Thus, indeed, the relation $R$ is left-closed.
\end{proof}

\begin{lemma}\label{lem:productEpsilon2}
  Let $\calL\subseteq\traceMon$ be rational. Then the relation
  \[
    \calR=(\{[\varepsilon]\times\calL\})\cdot\Id_{\traceMon}
    =\{(y,xy)\mid x\in\calL,y\in\traceMon\}
  \]
  is lc-rational.

  More precisely, from a dependence alphabet $\structD=(A,D)$ and an
  NFA $\autB=(Q,A,I,T,F)$, one can compute a left-closed transducer
  $\autT$ with
  $[R(\autT)]=(\{[\varepsilon]\}\times[L(\autB)]\})\cdot\Id_{\traceMon}$;
  this computation can be carried out in time polynomial in the size
  $\|\autB\|$ of the NFA $\autB$.
\end{lemma}

\begin{proof}
  From $\autB$, one can compute in polynomial time a transducer
  $\autT_1$ with $R(\autT_1)=\{\varepsilon\}\times L(\autB)$ by
  \cref{prop:wordsDirectProducts}(i). Let $\autT_2$ be a transducer
  for $\Id_{A^*}$. From the transducers $\autT_1$ and $\autT_2$, one
  can compute a transducer $\autT$ with
  $R(\autT)=(\{\varepsilon\}\times L(\autB))\cdot\Id_{A^*}$ in
  polynomial time (just add $(\varepsilon,\varepsilon)$-transitions
  from any accepting state of $\autT_1$ to any initial state of
  $\autT_2$).
  
  We first show that $R(\autT)$ is even left-closed. So let
  $u\in L(\autB)$ and $v\sim v'$ be arbitrary words such that
  $v\sim v'\relation{R(\autT)}uv'$. Then we have
  $v\relation{R(\autT)}uv\sim uv'$. Hence, indeed, $R(\autT)$ is
  left-closed implying that $\autT$ is left-closed.
  
  Next we show
  $(\{[\varepsilon]\}\times
  [L(\autB)])\cdot\Id_{\traceMon}=[R(\autT)]$. For the inclusion
  ``$\supseteq$'', let $(v,uv)\in R(\autT)$, i.e., $u\in L(\autB)$ and
  $v\in A^*$. Then $([v],[uv])=([\varepsilon],[u])\cdot([v],[v])$ is
  contained in the left-hand side, i.e, we showed the inclusion
  ``$\supseteq$''. Conversely, let
  $(y,xy)=([\varepsilon],x)\cdot(y,y)$ belong to the left-hand side,
  i.e., $x\in[L(\autB)]$ and $y\in\traceMon$. From $x\in[L(\autB)]$, we
  obtain a word $u\in L(\autB)$ with $[u]=x$. Further, there is a word
  $v\in A^*$ with $[v]=y$. It follows that $(v,uv)\in R(\autT)$ and therefore
  $(y,xy)=([v],[uv])\in[R(\autT)]$.
\end{proof}

Now the following sufficient condition for the lc-rationality of
$\calR_1\cdot\calR_2$ follows:

\begin{theorem}\label{thm:product}
  Let $\calK\subseteq\traceMon$ be recognizable,
  $\calL\subseteq\traceMon$ rational, and $\calR\subseteq\traceMon^2$
  lc-rational. Then $(\calK\times\calL)\cdot\calR$ is lc-rational.

  More precisely, from a dependence alphabet $\structD=(A,D)$, a
  closed NFA $\autA$, an NFA $\autB$, and a left-closed transducer
  $\autT$, one can compute a left-closed transducer $\autT'$ with
  $[R(\autT')]=([L(\autA)]\times [L(\autB)])\cdot [R(\autT)]$; this
  computation can be carried out in time polynomial in
  $\|\autA\|+\|\autB\|+\|\autT\|+\settwinindex(\structD)$, i.e., in the size of
  the NFAs,  the transducer, and the set twin index $\settwinindex(\structD)$
  of the dependence alphabet~$\structD$.
\end{theorem}

\begin{proof}
  In the following, let $\calK=[L(\autA)]$, $\calL=[L(\autB)]$, and
  $\calR=[R(\autT)]$.

  By \cref{lem:productEpsilon1}, one can compute in time polynomial in
  the size of the closed NFA $\autA$ and the set twin index of
  $\structD$ a left-closed transducer $\autT_1$ for the lc-rational
  relation $\calR_1=\calK\times\{[\varepsilon]\}$.
  \cref{lem:productEpsilon2} allows to construct in time polynomial in
  the size of the NFA $\autB$ a left-closed transducer $\autT_2$ for
  the lc-rational relation
  $\calR_2=\{[\varepsilon]\}\times\calL$. From
  \cref{prop:composition}(ii), we obtain a left-closed transducer
  $\autT'$ for the relation $\calR_1\circ\calR\circ\calR_2$ in time
  polynomial in the size of the transducers $\autT_1$, $\autT$, and
  $\autT_2$. Hence, in summary, the construction of $\autT'$ can be
  carried out in the given time bound.

  It remains to be shown that the relation $[R(\autT')]$ equals
  $(\calK\times\calL)\cdot\calR$, i.e., that
  $\calR_1\circ\calR\circ\calR_2=(\calK\times\calL)\cdot\calR$ holds.
  
  Note that $(x,z)\in\calR_1\circ\calR\circ\calR_2$ iff there are
  $y_1,y_2\in\traceMon$ with $(x,y_1)\in \calR_1$, $(y_1,y_2)\in\calR$, and
  $(y_2,z)\in\calR_2$. But $(x,y_1)\in\calR_1$ is equivalent to the existence of
  $k\in\calK$ with $x=k\cdot y_1$. Similarly, $(y_2,z)\in\calR_2$ iff there is
  $\ell\in\calL$ with $z=\ell\cdot y_2$. In summary, we have
  $(x,z)\in\calR_1\circ\calR\circ\calR_2$ iff there exist $k\in\calK$,
  $(y_1,y_2)\in\calR$, and $\ell\in\calL$ with
  $(x,z)=(k\, y_1,\ell\, y_2)$. But this holds iff
  $(x,z)\in(\calK\times\calL)\cdot\calR$.
\end{proof}

\subsection{Preservation of language properties}

Recall that rational word relations preserve the regularity of
languages under left- and right-application. Since rationality and
recognizability are different notions in the trace monoid, this leads
to two possible generalisations; later, \cref{lem:nonPreservation}
will show that none of them holds for rational trace relations (but,
by Theorem~\ref{thm:transformRational}, the right-application of a
rational trace relation transforms a recognizable trace language into
a rational one).

Now restrict attention to lc-rational trace relations $\calR$. Since
$\inv{\calR}$ need not be lc-rational, we now get four
possible preservation results: we could consider rationality or
recognizability as well as left- or right-application. The
following theorem shows that two of them hold,
\cref{lem:nonPreservation} proves that the other two fail.

\begin{theorem}\label{thm:preservation}
  Let $\calR\subseteq\traceMon\times\traceMon$ be an lc-rational trace
  relation.
  \begin{enumerate}[label=(\roman*)]
  \item If $\calK\subseteq\traceMon$ is recognizable, then also
    $^\calR\calK$ is recognizable.

    More precisely, from a dependence alphabet $\structD$, a
    left-closed transducer $\autT$, and a closed NFA $\autA$, one can
    compute in polynomial time a closed NFA $\autB$ with
    $[L(\autB)]={^{[R(\autT)]} [L(\autA)]}$.
  \item If $\calL\subseteq\traceMon$ is rational, then also $\calL^\calR$ is
    rational.

    More precisely, from a dependence alphabet $\structD$, a
    left-closed transducer $\autT$, and an NFA $\autA$, one can
    compute in polynomial time an NFA $\autB$ with
    $[L(\autB)]=[L(\autA)]^{[R(\autT)]}$.
  \end{enumerate}
\end{theorem}

\begin{proof}
  First, let $\autT$ be a left-closed transducer and $\autA$ a closed
  NFA. By \cref{prop:directProductTraces}(b), we can construct in
  polynomial time a left-closed transducer $\autT_1$ with
  $[R(\autT_1)]=[L(\autA)]\times\{[\varepsilon]\}$. In a second step,
  using \cref{prop:composition}(ii), we can construct in polynomial
  time a left-closed transducer $\autT_2$ with
  $[R(\autT_2)]=[R(\autT)]\circ[R(\autT_1)]$. Finally, by
  \cref{prop:directProductTraces}(a), we can construct in polynomial
  time a closed NFA $\autB$ with
  $[L(\autB)]=\{[v]\mid\exists u\colon([u],[v])\in[R(\autT_2)]\}$.

  This finishes the proof of the first claim since
  \[
    [R(\autT_2)] = [R(\autT)]\circ [R(\autT_1)] = 
    [R(\autT)]\circ \bigl([L(\autA)]\times\{[\varepsilon]\}\bigr) =
    {^{[R(\autT)]}[L(\autA)]}\times\{[\varepsilon]\}
  \]
  and therefore
  \[
    [L(\autB)]={^{[R(\autT)]}}[L(\autA)]\,.
  \]

  The proof of the second claim is analogous.
\end{proof}

We now come to the announced non-preservation results; they hold for
lc-rational trace relations and therefore, in particular, for the
larger class of rational trace relations.

\begin{lemma}\label{lem:nonPreservation}
  Suppose there are $a,b,c,d\in A$ with $(a,b)\in D$ and
  $c\parallel d$.
  
  There exist an lc-rational relation $\calR\subseteq\traceMon^2$, a
  rational set $\calK\subseteq\traceMon$, and a recognizable set
  $\calL\subseteq\traceMon$ such that ${^\calR}\calK$ is not rational and
  $\calL^{\calR}$ is not recognizable.
\end{lemma}

(We refer, again, to Theorem~\ref{thm:transformRational} expressing
that the left- or right-application of a rational trace relation
transforms a recognizable trace language into a rational one.)

\begin{proof}
  Let $R=\{(a,c),(b,d)\}^*$. Then $R$ is rational and, since
  $(a,b)\in D$, even lc-rational. We consider the lc-rational trace
  relation $\calR=[R]$. Since $c\parallel d$, we obtain
  $([u],[v])\in\calR$ if, and only if, $u\in\{a,b\}^*$, $v\in\{c,d\}^*$,
  $|u|_a=|v|_c$, and $|u|_b=|v|_d$.
  
  Consider the regular language $K=\{cd\}^*$ and let $\calK$ denote the
  rational set $[K]$. Since $c\parallel d$, we get $[v]\in\calK$
  iff $v\in\{c,d\}^*$ and $|v|_c=|v|_d$. It follows that
  $[u]\in{^\calR}\calK$ iff $u\in\{a,b\}^*$ and $|u|_a=|u|_b$. Let
  $H\subseteq A^*$ denote the set of words $u\in\{a,b\}^*$ with
  $|u|_a=|u|_b$. Since $(a,b)\in D$, this language $H$ is the only
  language with $[H]={^\calR}\calK$. Since $H$ is not regular, it
  follows that ${^\calR}\calK$ is not rational which proves the first
  claim.
  
  Next, let $L=(ab)^*$ and $\calL=[L]$. Then $[u]\in \calL$ iff $u\in L$
  since $(a,b)\in D$. Hence $\{u\in A^*\mid [u]\in\calL\}$ is the
  regular language $L$, implying that $\calL$ is recognizable. Note
  that $\calL^\calR$ is the set of traces $[v]$ with $v\in\{c,d\}^*$ and
  $|v|_c=|v|_d$ (i.e., it equals $\calK$). Hence the language
  $\{v\in A^*\mid [v]\in\calL^\calR\}$, i.e., $\calL^\calR$ is not
  recognizable.
\end{proof}

The above lemma implies, in particular, that the right-application of
rational trace relations does neither preserve the rationality nor the
recognizability of a trace language. Proposition~\ref{prop:ratRelsplit}
allows to prove the weaker result that the right-application of a
rational relation to a recognizable set yields a rational set.

\begin{theorem}\label{thm:transformRational}
  Let $\calR\subseteq\traceMon\times\traceMon$ be rational and
  $\calL\subseteq\traceMon$ recognizable. Then $\calL^\calR$ is rational.
\end{theorem}

\begin{proof}
  By \cref{prop:ratRelsplit}, there are lc-rational relations
  $\calR_1$ and $\calR_2$ such that
  $\calR=\inv{\calR_1}\circ\calR_2$. Hence
  $\calL^\calR=\calL^{\inv{\calR_1}\circ\calR_2} =
  (\calL^{\inv{\calR_1}})^{\calR_2} = (^{\calR_1}\calL)^{\calR_2}$.
  From \cref{thm:preservation}, we know that
  $\calK:={^{\calR_1}\calL}$ is recognizable and therefore in
  particular rational. Hence, again using \cref{thm:preservation},
  $\calK^{\calR_2}$ is rational.
\end{proof}

Note that the above lemma cannot be improved: by
\cref{lem:nonPreservation}, there exist an lc-rational relation
$\calR$ and a recognizable set $\calL$ such that $\calL^\calR$ is not
recognizable. Furthermore, there is also a rational set $\calL$ such
that $^\calR\calL$ is not rational. Note that
$^\calR\calL=\calL^{\inv{\calR}}$ and that $\inv{\calR}$ is
rational. Hence, indeed, the above lemma is optimal.

%%%%%%%%%%%%%%%%%%%%%%%%%%%%%%%%%%%%%%%%%%%%%%%%%%%%%%%%%%%%%%%%%%%%%%
\section{The Reachability Relation of tPDS is LC-Rational}
\label{sec:reachRel}

In this section we consider the reachability relation of
trace-pushdown systems. Concretely, we will show that, by application
of the results from the previous section, this relation is
lc-rational. Recall that $\vdash^*_\autP$ denotes the reachability
relation of the tPDS $\autP=(Q,\Delta)$, i.e., the set of pairs
$(c,d)$ of configurations such that $d$ is reachable from $c$.

Next, we want to introduce the notions of rationality and
lc-rationality for the reachability relation. Note that we cannot
derive them from the classical definitions presented in
\cref{sec:prelim} since the set of configurations $\Conf{\autP}$ is
not a monoid. However, when fixing any pair of states $p,q\in Q$ and
projecting to the pushdown, we obtain the trace relation
\[
  \reach_{p,q}(\autP):=\{([u],[v])\in\traceMon^2
  \mid (p,[u])\vdash_\autP^*(q,[v])\}\,.
\]
We say that the reachability relation of $\autP$ is
\emph{(lc-)rational} if, and only if, for each pair
$p,q\in Q$ of states, the trace relation
$\reach_{p,q}(\autP)$ is (lc-)rational.

The main theorem of this section shows that this holds for
any trace-pushdown system:

\begin{theorem}\label{thm:reachLCRational}
  Let $\structD=(A,D)$ be a dependence alphabet and $\autP=(Q,\Delta)$
  a trace-pushdown system. Then the reachability relation
  $\vdash^*_\autP$ is lc-rational.

  More precisely, from $\structD$, $\autP$, and $p,q\in Q$, one can
  construct in time polynomial in
  $||\autP||^{\mathrm{poly}(\twinindex(\structD))}$ a left-closed
  transducer $\autT$ such that $[R(\autT)]=\reach_{p,q}(\autP)$.
\end{theorem}

Suppose $D=A\times A$. Then $\twinindex(\structD)=1$
and the trace-pushdown system $\autP$ is actually a classical pushdown
system. Furthermore, in this case, the transducer $\autT$ can be
constructed in time polynomial in the size of the pushdown system
$\autP$---thus, the above theorem generalizes the classical results
from~\cite{Cau88,FinWW97}.

The proof of this theorem (that can be found on page
\pageref{sec:proof-forwards}) is inspired by the work by Finkel et
al.\ \cite{FinWW97}.\footnote{The result originates from \cite{Cau88}
  where Caucal demonstrates that the reachability relation of a
  pushdown system is effectively prefix-recognizable, i.e., a finite
  union of sets of the form $\{(uw,vw)\mid u\in U, v\in V, w\in A^*\}$
  for regular languages $U$ and $V$, but our construction generalizes
  the one from \cite{FinWW97}.} To explain its idea and
particularities, we first start with a classical pushdown system
$\autP=(Q,\Delta)$ (i.e., with a trace-pushdown system over the trace
monoid $A^*$). Suppose there are states $p$, $q$, and $r$ and
transitions $(p,a,bv,q)$ and $(q,b,\varepsilon,r)$ such that
$(p,ax) \vdash_\autP (q,bvx)\vdash_\autP (r,vx)$ holds for any word
$x$. If we add the transition $(q,a,v,r)$ to $\Delta$ that allows to
go from $(p,ax)$ to $(r,vx)$ in one step, the reachability relation
does not change. We keep adding such ``shortcuts'' and call the
resulting set of transitions $\Delta^{(\infty)}$. Note that this set
of transitions is finite since any added ``shortcut'' writes a word
that is shorter than the word $bv$ from the transition $(p,a,bv,q)$
above. Further, the reachability relations of $\autP$ and of
$\autP^{(\infty)}=(Q,\Delta^{(\infty)})$ are the same. We next split
$\Delta^{(\infty)}$ into the set of transitions shortening the
pushdown and the set of transitions that do not shorten the pushdown:
\begin{align*}
  \Delta^{(\infty)}_\varepsilon &= \{(p,a,\varepsilon,q)\in\Delta^{(\infty)}\}
  & \autP^{(\infty)}_\varepsilon &= (Q,\Delta^{(\infty)}_\varepsilon)\\
  \Delta^{(\infty)}_+ &= \{(p,a,v,q)\in\Delta^{(\infty)}\colon |v|\ge1\}
  & \autP^{(\infty)}_+ &= (Q,\Delta^{(\infty)}_+)\,.
\end{align*}
The crucial point of the arguments by Finkel et al.\ is the following:
for any two configurations $(p,u)$ and $(r,w)$, we then get
$(p,u)\vdash_\autP^*(r,w)$ if, and only if,
$(p,u)\vdash_{\autP^{(\infty)}}^*(r,w)$ if, and only if, there exists
a configuration $(q,v)$ such that
$(p,u)\vdash_{\autP^{(\infty)}_\varepsilon}^*(q,v)
\vdash_{\autP^{(\infty)}_+}^* (r,w)$. In other words, any run of the
original system $\autP$ can be simulated by a run of the system with
shortcuts $\autP^{(\infty)}$ that first shortens the pushdown (using
transitions from $\Delta^{(\infty)}_\varepsilon$) and then writes onto
the pushdown (using transitions from $\Delta^{(\infty)}_+$).

It follows that, for any set of configurations $C$, we have
\[
  \reach_{p,r}(\autP)=\bigcup_{q\in Q}
  \reach_{p,q}(\autP_\varepsilon^{(\infty)})\circ
  \reach_{q,r}(\autP_+^{(\infty)})\,.
\]

Due to the very restricted type of transitions in the two subsystems
of $\autP^{(\infty)}$, it follows that
$\reach_{p,q}(\autP_\varepsilon^{(\infty)})$ and
$\reach_{q,r}(\autP_+^{(\infty)})$ are rational word relations.
Consequently, by \cref{thm:ratWordRel}(R2) the reachability relation
$\vdash^*_\autP$ is rational for pushdown systems $\autP$.

The crucial point of the above construction is that any run of the
system $\autP^{(\infty)}$ can be brought into some ``simple form'' by
using shortcuts. Here, ``simple form'' means that it consists of two
phases: the pushdown decreases properly in every step of the first
phase and does not decrease in any step of the second phase.

Our strategy in the proof of \cref{thm:reachLCRational} will extend
the above idea:\label{page:strategy}
\begin{enumerate}
\item First, Propositions~\ref{prop:read} and \ref{prop:write}
  demonstrate that \cref{thm:reachLCRational} holds for
  ``homogeneous'' systems that formalize and strengthen the two types
  of systems $\autP^{(\infty)}_\varepsilon$ and $\autP^{(\infty)}_+$
  from above:
  
  \begin{definition}
    Let $\autP=(Q,\Delta)$ be a trace-pushdown system. It
    is \emph{homogeneous} if one of the following hold:
    \begin{enumerate}[label=(\arabic*)]
      \item $\Delta\subseteq Q\times A\times\{\varepsilon\}\times Q$, i.e.,
        all transitions $(p,a,w,q)\in\Delta$ satisfy $w=\varepsilon$, or
      \item $\Delta\subseteq Q\times\twins(a)\times A^+\times Q$ for a
        letter $a\in A$, i.e., all transitions $(p,b,w,q)\in\Delta$
        satisfy $D(a)=D(b)$ and $w\neq\varepsilon$.
    \end{enumerate}
  \end{definition}
  Note that the set of transitions $\Delta$ of any trace-pushdown
  system $\autP$ can be split into the set of transitions
  $\Delta_\varepsilon$ as above, and the sets
  $\Delta_{\twins(b)}=\{(p,a,w,q)\in\Delta\mid D(a)=D(b)\text{ and
  }w\neq\varepsilon\}$. Hence, the number of these subsystems (that
  was $2$ above) is linear in the twin index $\twinindex(\structD)$ of
  the dependence alphabet $\structD$.
\item Secondly, \cref{thm:reachSaturated} demonstrates
  \cref{thm:reachLCRational} for ``saturated'' systems, i.e., systems
  where no new ``shortcuts'' can be added:
  \begin{definition}
    Let $\autP=(Q,\Delta)$ be a trace-pushdown system. It
    is \emph{saturated} if $(p,a,ubv,q),(q,b,\varepsilon,r)\in\Delta$
    with $u\parallel b$ implies $(p,a,uv,r)\in\Delta$.
  \end{definition}
  The central argument in this proof is that any run can be
  transformed into an equivalent one that consists of a bounded number
  of phases. As explained above, this bound is~2 for pushdown systems,
  \cref{ex:twophases-counterexample} will show that this small bound
  does not suffice for trace-pushdown systems over trace monoids other
  than $A^*$. But we show that this number of phases is linear in the
  twin index $\twinindex(\structD)$ of $\structD$. This number of phases is
  the reason why our procedure is exponential in $\twinindex(\structD)$.
\item Finally, \cref{prop:saturation} proves
  \cref{thm:reachLCRational} in full generality by showing that any
  system can be saturated by adding shortcuts.
\end{enumerate}
We will do the aforementioned steps in the following three
subsections.

\subsection{Reachability in homogeneous systems}
\label{ssec:step2}

We first consider the reachability relation for systems that either
shorten their pushdown in each transition or that do not shorten it in
any transition, but only replace letters $a$ from the pushdown with
$D(a)=D(b)$ for some fixed letter $b$. Accordingly, we prove two
propositions.

The first result considers systems that shorten their pushdown in every
step.
\begin{proposition}\label{prop:read}
  Let $\structD=(A,D)$ be a dependence alphabet, $\autP=(Q,\Delta)$ 
  a trace-pushdown system with
  $\Delta\subseteq Q\times A\times\{\varepsilon\}\times Q$, and
  $p,q\in Q$ two states. Then the relation $\reach_{p,q}(\autP)$ is
  lc-rational.

  Even more, from $\structD$, $\autP$, and $p,q\in Q$, one can compute a
  left-closed transducer $\autT$ with
  $[R(\autT)]=\reach_{p,q}(\autP)$; this computation can be carried
  out in time polynomial in $\|\autP\|+\settwinindex(\structD)$.
\end{proposition}

\begin{proof}
  We first transform the tPDS $\autP$ into the NFA
  $\autA=(Q,A,\{p\},T,\{q\})$ setting
  \[
    (p_1,a,p_2)\in T\iff(p_1,a,\varepsilon,p_2)\in\Delta\,.
  \]
  Note that the tPDS $\autP$ only reads letters from the pushdown and
  never writes anything onto the pushdown.  Essentially, $\autA$ is
  the tPDS $\autP$, but we read letters from the input instead of
  the stack.

  Let $\calK=[L(\autA)]$. We prove that $\calK$ is recognizable and that
  \[
    \reach_{p,q}(\autP)=\left(\calK\times\{[\varepsilon]\}\right) \cdot
    \Id_{\traceMon}\,.
  \]

  To prove the recognizability of $\calK$, it suffices to prove
  \begin{equation}
    u\in L(\autA)\iff[u]\in\calK \label{eq:prfRead}
  \end{equation}
  for all $u\in A^*$ since $L(\autA)$ is regular (note that this
  implies in particular that the NFA $\autA$ is closed). The
  implication ``$\Rightarrow$'' is immediate by the definition of the
  trace language $\calK$.  For the implication ``$\Leftarrow$'', let
  $[u]\in\calK$. By the definition of $\calK$, there exists
  $v\in L(\autA)$ with $[u]=[v]$, i.e., $u\sim v$. Hence $u$ can be
  obtained from $v$ by transposing consecutive independent
  letters. Since the tPDS $\autP$ satisfies the diamond property
  \ref{def:diamond1}, the NFA $\autA$ satisfies the diamond property
  \ref{def:diamond1-DFA}. Hence, the $v$-labeled path in $\autA$ from
  $p$ to $q$ can be transformed into a $u$-labeled path from $p$ to
  $q$. But this means that the word $u$ is accepted by the NFA
  $\autA$, i.e., $u\in L(\autA)$.

  Thus, indeed, the trace language $\calK$ is recognizable.

  We now prove the above characterization of the relation
  $\reach_{p,q}(\autP)$.  So let $u,v\in A^*$. Then
  $([u],[v])\in\reach_{p,q}(\autP)$ iff $(p,[u])\vdash^*(q,[v])$. But
  this is equivalent to the existence of $n\ge0$, transitions
  $(p_{i-1},a_i,\varepsilon,p_i)\in\Delta$, and words $x_i$ for
  $1\le i\le n$ such that
  \begin{enumerate}
  \item $(p,[u])=(p_0,[a_1x_1])$,
  \item $x_i\sim a_{i+1} x_{i+1}$ for all $1\le i<n$, and
  \item $(p_n,[x_n])=(q,[v])$.
  \end{enumerate}
  Now suppose these transitions and words exist. The construction of
  the NFA $\autA$ yields $(p_{i-1},a_i,p_i)\in T$ for all
  $1\le i\le n$. Hence
  \[
    p=p_0\xrightarrow{\ a_1\ }p_1\xrightarrow{\ a_2\ }p_2\cdots \xrightarrow{\ a_n\ }p_n=q
  \]
  is a path in $\autA$ implying $a_1\cdots a_n\in L(\autA)$ and
  therefore $[a_1\cdots a_n]\in\calK$. Recall that we also have
  $[x_n]=[v]$ and, by induction,
  $[x_i]=[a_{i+1} a_{i+2}\cdots a_n v]$, in particular
  $u\sim a_1 x_1\sim a_1a_2\cdots a_n v$. Consequently,
  \[
    ([u],[v])=([a_1\cdots a_n],[\varepsilon])\cdot ([v],[v])
    \in(\calK\times\{[\varepsilon]\})\cdot \Id_{\traceMon}
  \]
  which proves the first inclusion.

  Conversely, suppose $[u]\in\calK$ and $v\in A^*$, i.e.,
  $([u],[\varepsilon])\cdot([v],[v])\in
  (\calK\times\{[\varepsilon]\})\cdot\Id_{\traceMon}$. Since
  $[u]\in\calK$, we obtain from Equation~\eqref{eq:prfRead} a
  $u$-labeled accepting path, say
  \[
    p=p_0\xrightarrow{\ a_1\ }p_1\xrightarrow{\ a_2\ }p_2\cdots \xrightarrow{\ a_n\ }p_n=q
  \]
  is such a path. The construction of the NFA $\autA$ yields
  $(p_{i-1},a_i,\varepsilon,p_i)\in\Delta$ for all $1\le i\le n$. With
  $x_i=a_i a_{i+1}\cdots a_n\, v$, there are also words such that the
  above three properties hold, implying
  $([u],\varepsilon)\cdot([v],[v])\in\reach_{p,q}(\autP)$ and
  therefore the converse inclusion.

  Since the closed NFA $\autA$ can be computed in time polynomial in
  the size of $\autP$, the claim follows from
  Lemma~\ref{lem:productEpsilon1}.
\end{proof}

In the above proof, we constructed a recognizable trace language
$\calK$, proved that the trace relation $\reach_{p,q}(\autP)$ equals
$(\calK\times\{[\varepsilon]\})\cdot\Id_{\traceMon}$ (this part of the
proof used that $\autP$ can only shorten its pushdown) and used that
such relations are lc-rational.

We next consider systems where no transition shortens the
pushdown. Here, we make the additional assumption that all transitions
replace letters $b$ with $D(a)=D(b)$ for some fixed letter $a$.  In
this situation, we analogously to the above proof construct a rational
trace language $\calL$, prove that the trace relation
$\reach_{p,q}(\autP)$ equals
$(\{[\varepsilon]\}\times\calL)\cdot\Id_{\traceMon}$, and use that
such relations are lc-rational.

Recall that $\twins(a)$ is the set of twins of the letter $a$, hence
we restrict to systems where all transitions replace twins of $a$.

\begin{lemma}\label{lem:write1}
  Let $\structD=(A,D)$ be a dependence alphabet, $\autP=(Q,\Delta)$ a
  trace-pushdown system with
  $\Delta\subseteq Q\times \twins(a)\times A^+\times Q$ for some
  $a\in A$, and $p,q\in Q$ be two states.  There exists a regular
  language $H_a\subseteq A^*$ such that, for any $v\in A^*$,
  \[
    \{[w]\mid
    (p,[av])\vdash_{\autP}^*(q,[w])\}=[H_a]\cdot \{[v]\}\,.
  \]

  More precisely, from $\autP$, $a\in A$, and $p,q\in Q$, one can
  construct in time polynomial in the size of $\autP$ an NFA $\autA_a$
  with
  \[
    \{[w]\mid(p,[av])\vdash^*_\autP(q,[w])\}=[L(\autA_a)]\cdot\{[v]\}
  \]
  for all $v\in A^*$.
\end{lemma}

\begin{proof}
  Let $\autA=(Q_\autA,A,I,\delta,F)$ be the following
  $\varepsilon$-NFA: we start with the only initial state
  $(q,\varepsilon)$ and all pairs $(r,c)\in Q\times A$ as further
  states (recall that $Q$ is the set of states of the trace-pushdown
  system $\autP$). The idea is that the second component stores the
  top letter of the pushdown that is replaced in the next step (if
  such a step exists).  To start, we add $c$-labeled transitions from
  $(q,\varepsilon)$ to $(q,c)$ for any letter $c\in A$. Then, for any
  transition $(r,c,udv,s)\in\Delta$ with $d\in A$ and $u,v\in A^*$
  such that $u\parallel d$, we add a $uv$-labeled path from $(s,d)$ to
  $(r,c)$. The set of initial states is $I=\{(q,\varepsilon)\}$ and
  the set of final states is $F=\{(p,a)\}$.
  
  We set $H_a=L(\autA)$.

  Now let $v\in A^*$ be arbitrary and set
  $\calL_v=\{[w]\mid (p,[av])\vdash^*_\autP(q,[w])\}$. Hence, it
  remains to be shown that $\calL_v=[H_a]\cdot\{[v]\}$.
  
  First, we verify the inclusion ``$\supseteq$''. Therefore, let
  $x\in[H_a]\cdot\{[v]\}$ be arbitary. Then there is $u\in H_a$ with
  $x=[u]\cdot[v]=[uv]$. Since $u$ is accepted by the $\varepsilon$-NFA
  $\autA$, we find a letter $a_0\in A$, pairs
  $(a_j,w_j)\in\twins(a)\times A^*$ for $j\in[n]$, and tPDS-states
  $q_j\in Q$ for $j\in\{0,1,\dots,n\}$ such that
  $u=a_0 w_1 w_2 \cdots w_n$ as well as
  \begin{itemize}
    \item $q=q_0$,
    \item there is a $w_{j+1}$-labeled path from $(q_j,a_j)$ to
    $(q_{j+1},a_{j+1})$ that does not contain any inner state
    from $Q\times\twins(a)$, and
    \item $(q_n,a_n)=(p,a)$.
  \end{itemize}
  The construction of the $\varepsilon$-NFA $\autA$ implies that
  for all $0\le j<n$ there are transitions
  $(q_{j+1},a_{j+1},u_{j+1} a_j v_{j+1},q_j)\in \Delta$ with
  $u_{j+1}\parallel a_j$ and $w_{j+1}=u_{j+1} v_{j+1}$. Hence, we have
  \begin{align*}
    (p,[av]) = (q_n,[a_n v])
    &\vdash_\autP (q_{n-1},[u_n a_{n-1} v_n v])\\
    &\phantom{\vdash_\autP{}}=(q_{n-1},[a_{n-1} u_n v_n v])
      =(q_{n-1},[a_{n-1} w_n v])\\
    &\vdash_\autP (q_{n-2},[u_{n-1} a_{n-2} v_{n-1} w_n v])\\
    &\phantom{\vdash_\autP{}} =(q_{n-2},[a_{n-2} u_{n-1} v_{n-1} w_n v])
      =(q_{n-2},[a_{n-2} w_{n-1} w_n v])\\
    &\ \ \vdots\\
    &\vdash_\autP (q_0,[a_0 w_1 w_2\cdots w_n v])=(q,[uv])\,.
  \end{align*}
  Consequently, $x=[u]\cdot[v]=[uv]\in \calL_v$. Since
  $x\in [H_a]\cdot\{[v]\}$ was chosen arbitrary, we have
  $\calL_v\supseteq [H_a]\cdot\{[v]\}$.
  
  For the converse inclusion, let $w\in A^*$ with $[w]\in\calL_v$,
  i.e., $(p,[av])\vdash_\autP^* (q,[w])$. In order to prove
  $[w]\in[H_a]\cdot\{[v]\}$, we will construct $u\in H_a$ such that
  $w\sim uv$.
  
  From $(p,[av])\vdash_\autP^* (q,[w])$, we get a
  natural number $n\ge0$, states $q_j\in Q$, and words $x_j\in A^*$
  for all $j\in\{0,1,\dots,n\}$ such that
  \begin{align*}
    (p,[av]) = (q_n,[x_n])
    & \vdash_\autP (q_{n-1},[x_{n-1}])\\
    & \vdash_\autP (q_{n-2},[x_{n-2}])\\
    & \ \ \vdots\\
    & \vdash_\autP (q_1,[x_1])\\
    & \vdash_\autP (q_0,[x_0]) = (q,[w])\,.
  \end{align*}
  Consequently, for any $j$ with $n\ge j>0$, there is a transition
  $(q_j,a_j,u_j',q_{j-1})\in\Delta$ and a word $y_j\in A^*$ with
  $x_j\sim a_j y_j$ and $x_{j-1}\sim u_j' y_j$ (i.e., the
  trace-pushdown system replaces the letter $a_j$ with the trace
  $[u_j']$).  The requirement on $\autP$ implies $a_j\in \twins(a)$ and
  $u_j'\neq\varepsilon$, hence there are $b_{j-1}\in A$ and
  $u_j\in A^*$ with $u_j'=b_{j-1}u_j$. Note that
  $b_{j-1}\in D(b_{j-1})\subseteq D(u_j')\subseteq D(a_j)$ since
  $\autP$ is a trace-pushdown system.

  Now let $n\ge j>1$. Then we obtain
  \[
    b_{j-1} u_j y_j = u_j' y_j \sim x_{j-1} \sim a_{j-1} y_{j-1}\,.
  \]
  From $b_{j-1}\in D(a_j)=D(a)=D(a_{j-1})$, we infer $b_{j-1}=a_{j-1}$
  as well es $u_j y_j \sim y_{j-1}$.

  Consequently, we have
  \begin{align*}
    x_0\sim u_1'y_1 = b_0 u_1 y_1
    &\sim b_0 u_1 u_2 y_3\\
    &\ \ \vdots\\
    &\sim b_0 u_1 u_2 \cdots u_n y_n
  \end{align*}
  Next recall $av\sim x_n\sim a_n y_n$ and $a\in D(a)=D(a_n)$. Hence
  $a=a_n$ and $v\sim y_n$. Thus, we get
  \[
    w\sim x_0\sim b_0 u_1 u_2\cdots u_n y_n \sim
    b_0 u_1 u_2\cdots u_n v = u v
  \]
  with $u:=b_0 u_1 u_2\cdots u_n$.
  
  Finally consider the following path in the $\varepsilon$-NFA $\autA$:
  \[
  (q,\varepsilon)=(q_0,\varepsilon)
  \xrightarrow{b_0}_{\autA} (q_0,b_0)
  \xrightarrow{u_1}_{\autA} (q_1,a_1)
  \cdots
  \xrightarrow{u_n}_{\autA} (q_n,a_n) = (q,a)\,.
  \]
  It witnesses that the word $u$ is accepted by the $\varepsilon$-NFA
  $\autA$, i.e., that $u\in H_a$.

  Thus, indeed, we found a word $u\in H_a$ such that $w\sim uv$ and
  therefore $[w]\in[H_a]\cdot\{[v]\}$.
  
  Since $w\in A^*$ with $[w]\in\calL_v$ was chosen arbitrary, we also
  proved the inclusion $\calL_v\subseteq[H_a]\cdot\{[v]\}$.

  Finally note that the $\varepsilon$-NFA $\autA$ can be constructed
  in time polynomial in the size of the trace-pushdown system $\autP$
  and that $\autA$ can be transformed into an equivalent NFA $\autA_a$
  in time polynomial in the size of $\autA$.
\end{proof}

\begin{proposition}\label{prop:write}
  Let $\structD=(A,P)$ be a dependence alphabet, $\autP=(Q,\Delta)$ a
  trace-pushdown system with
  $\Delta\subseteq Q\times \twins(b)\times A^+\times Q$ for some
  $b\in A$, and $p,q\in Q$ two states. Then the trace relation
  $\reach_{p,q}(\autP)$ is lc-rational.

  More precisely, from $\structD$, $\autP$ and $p,q\in Q$, one can
  compute a left-closed transducer $\autT$ such that
  $[R(\autT)]=\reach_{p,q}(\autP)$; this computation can be carried
  out in time polynomial in $\|\autP\|+\settwinindex(\structD)$.
\end{proposition}

\begin{proof}
  For a letter $a\in \twins(b)$, let $H_a$ denote the regular
  set from Lemma~\ref{lem:write1}. Furthermore, set
  $\calI_{p,q}=\Id_{\traceMon}$ if $p=q$ and $\calI_{p,q}=\emptyset$
  if $p\neq q$ (note that $\calI_{p,q}$ is efficiently lc-rational).

  Now consider two words $v',w\in A^*$. Since all transitions from $\Delta$ read
  some letter from $\twins(b)$, we obtain $(p,[v'])\vdash^*_\autP(q,[w])$
  iff
  \begin{itemize}
  \item $p=q$ and $v'\sim w$ (i.e., $([v'],[w])\in\calI_{p,q}$) or
  \item there exist $a\in\twins(b)$ and $v\in A^*$ with $v'\sim av$ and
    $(p,[av])\vdash^*_\autP(q,[w])$.
  \end{itemize}
  By Lemma~\ref{lem:write1}, $(p,[av])\vdash^*_\autP(q,[w])$ holds iff
  there is $u\in H_a$ with $w\sim uv$. Thus,
  $(p,[v'])\vdash^*_\autP(q,[w])$ iff
  \[
    ([v'],[w])\in\calI_{p,q}\cup\bigcup_{a\in\twins(b)}\bigl(\{[a]\}\times[H_a]\bigr)\cdot\Id_{\traceMon}\,.
  \]
  In other words, this relation equals $\reach_{p,q}(\autP)$. Since
  $\{[a]\}$ is recognizable and $[H_a]$ rational, the relation
  $\reach_{p,q}(\autP)$ is indeed lc-rational by
  Theorem~\ref{thm:product}.

  More precisely, note that NFAs for $\{[a]\}$ and $H_a$ can be
  computed in polynomial time. Hence, by Theorem~\ref{thm:product}, a
  transducer for some left-closed relation $R_a$ with
  $[R_a]=(\{[a]\}\times[H_a])\cdot\Id_{\traceMon}$ can be computed in
  time polynomial in the size of the tPDS and the set twin index of
  $\structD$. The disjoint union of these transducers (for
  $a\in\twins(b)$) together with some transducer for $\calI_{p,q}$ is
  a left-closed transducer $\autT$ with
  $[R(\autT)]=\reach_{p,q}(\autP)$ and can be computed in the
  available time since $|A|\le\|\autP\|$.
\end{proof}

\subsection{Reachability in saturated systems}
\label{ssec:step3}

We saw that any run in a saturated pushdown system over the trace
monoid $A^*$ can be simulated by a run consisting of two phases:
first, the pushdown shortens and then, it increases. In this section,
we want to prove that a similar property holds for saturated
trace-pushdown systems, i.e., that every run can be simulated by a run
consisting of a uniformly bounded number (that only depends linearly
on $\twinindex(\structD)\le|A|$) of phases. By
\cref{prop:composition,prop:read,prop:write}, it then follows
immediately that the reachability relation of a saturated
trace-pushdown system is lc-rational. The following example shows
that, differently from the pushdown case, we cannot bound the number
of phases to two.

\begin{figure}[h]
  \begin{center}
    \begin{tikzpicture}[nfa]
      \node[state] (p) {$0$};
      \node[state,right of=p] (q) {$1$};
      \node[state,right of=q] (r) {$2$};
      \node[state,right of=r] (s) {$3$};
      
      \draw (p) edge[loop above] node {$a\mid abc$} (p)
                edge node[below] {$a\mid c$} (q);
      
      \draw (q) edge[loop above] node {$b\mid \varepsilon$} (p)
                edge node[below] {$b\mid \varepsilon$} (r) ;
      
      \draw (r) edge[loop above] node {$b\mid bde$} (r)
                edge node[below] {$b\mid e$} (s);
      \draw (s) edge[loop above] node{$d\mid \varepsilon$} (s);
    \end{tikzpicture}
  \end{center}
  \caption{The trace-pushdown system from Example
    \ref{ex:twophases-counterexample}.\label{fig:twophases-counterexample}}
\end{figure}
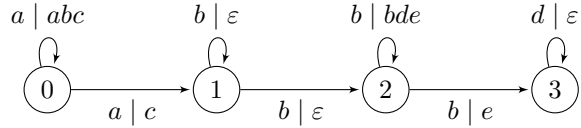

\begin{example}\label{ex:twophases-counterexample}
  We consider the dependence alphabet $\structD=(A,D)$ with
  $A=\{a,b,c,d,e\}$ and the dependence relation $D$ is given by the
  following table:
  \begin{center}
    \begin{tabular}{c|ccccc}
      $x$ & $a$ & $b$ & $c$ & $d$ & $e$\\\hline
      $D(x)$ & $A$ & $\{a,b,d,e\}$ & $\{a,c\}$ & $\{a,b,d\}$ & $\{a,b,e\}$
    \end{tabular}
  \end{center}
  Further, let $\autP=(Q,\Delta)$ be the trace-pushdown system from
  \cref{fig:twophases-counterexample}.  One can check that $\autP$ is
  saturated. The following is the only run from the configuration
  $(0,[a])$ to the configuration $(3,[e^4c^4])$:
  \begin{align*}
    (0,[a]) & \vdash^3 (0,[abcbcbc])
              \vdash (1,[cbcbcbc]) = (1,[b^3c^4]) \\
            & \vdash^2 (2,[bc^4])\\
            & \vdash^3 (2,[bdededec^4])
              \vdash (3,[edededec^4]) = (3,[d^3 e^4 c^4]) \\
            & \vdash^3 (3,[e^4c^4])
  \end{align*}
  Note that this run splits into four phases (that correspond to the
  four lines above); it increases its pushdowns in the first and third
  and decreases them in the second and fourth.
\end{example}

So far, we used the term ``phase'' without defining it formally. To be
 precise, a ``phase'' is a run of a maximal homogeneous
subsystem of $\autP$. These maximal homogeneous subsystems
$\autP_\varepsilon$ and $\autP_T$ for $T\in\twins(\structD)\subseteq 2^A$
are defined as follows:
\begin{align*}
  \Delta_\varepsilon &= \Delta\cap(Q\times A\times\{\varepsilon\}\times Q)
    & \autP_\varepsilon &= (Q,\Delta_\varepsilon)\\
  \Delta_T &= \Delta \cap (Q\times T\times A^+\times Q)
    & \autP_T &= (Q,\Delta_T)
\end{align*}
The set $\Delta_T$ contains all transitions that replace some letter
$a\in T$ by some non-empty word. Note that the definition of a homogeneous
subsystem of $\autP$ requires that $T=\twins(a)$ holds for some letter
$a\in A$. We will later use that the reachability relation in such
subsystems is efficiently lc-rational. We should also note, that there
is only a linear number of such maximal homogeneous subsystems $\autP_T$
since there are only $\twinindex(\structD)\leq|A|$ many sets
$T\in\twins(\structD)$.

Since $\twins(\structD)$ is a partitioning of $A$ into the equivalence
classes of the relation ``twin'', the set of transitions $\Delta$ is the
disjoint union of its subsets $\Delta_\varepsilon$ and $\Delta_T$ for
$T\in\twins(\structD)$. Therefore, any run of $\autP$ splits into maximal
subruns of these subsystems and these subruns are precisely what we called
``phase''. In particular, any run of the system $\autP_\varepsilon$ or
$\autP_T$ is a single phase of a run of the complete system $\autP$.

We write $\vdash_\varepsilon$ for the one-step relation
$\vdash_{\autP_\varepsilon}$ of the system $\autP_\varepsilon$;
$\vdash_T$ is to be understood similarly.

For two binary relations $\models_1$ and $\models_2$ on a set $C$, we
write $\models_1\models_2$ as shorthand for the composition
$\models_1\circ\models_2$.

\begin{definition}
  Let $\autP$ be a trace-pushdown system.
  \begin{itemize}
  \item For $T\in\twins(\structD)$, we set
    $\mathord{\Vdash_T}={\vdash^*_\varepsilon\,\vdash^+_T}\subseteq\Conf{\autP}
    \times \Conf{\autP}$.
  \item For $\ov{T}=(T_1,T_2,\dots,T_n)$ with $T_i\in\twins(\structD)$, we set
    $\Vdash_{\ov{T}} {=}
    \Vdash_{T_1}\,\Vdash_{T_2}\,\cdots\,\Vdash_{T_n}$.
  \end{itemize}
\end{definition}

In other words, a pair of configurations $(c_1,c_2)$ belongs to
$\Vdash_T$ if the system $\autP$ has a run from $c_1$ to $c_2$ that
first shortens the pushdown and then, in the second phase, uses
transitions from $\Delta_T$, only (that replace letters with $\twins(a)=T$
by non-empty words). It should be noted that the first (deleting)
phase is allowed to be empty while the second (writing) phase is
required to be non-empty. More generally, the pair $(c_1,c_2)$ belongs
to $\Vdash_{\ov{T}}$ if there is a run $r$ from $c_1$ to $c_2$ that
can be split into $r_1,r_2,\dots,r_{2n}$ where $r_{2i-1}$ shortens the
pushdown and $r_{2i}$ uses transitions from $\Delta_{T_i}$, only (for
all $1\le i\le n$); again, the odd subruns can be empty while the even
ones are required to be at least of length one.

Clearly, the binary relation $\vdash^*$ is the union of all relations
$\Vdash_{\ov{T}}\,\vdash_\varepsilon^*$ for $\ov{T}$ a sequence of
sets from $\twins(\structD)$ of arbitrary length. Our next aim is to
show that we only need to consider sequences $\ov{T}$ of bounded
length. To this aim, we will prove that any run witnessing
$c\Vdash_{\ov{T}} d$ for some ``long'' sequence $\ov{T}$ implies the
existence of some ``short'' sequence $\ov{U}$ such that
$c\Vdash_{\ov{U}} d$.

To see that the following lemma achieves this for certain runs,
suppose $u_0\neq\varepsilon$. Then using the condition \eqref{eq:swallow1}
from the lemma we infer in particular
$(p_0,[a_0x_0])\Vdash_{(\twins(a_0),T_1,\dots,T_n,T_{n+1})} d_3$. In the
very specific situation of the lemma, we then get
$(p_0,[a_0x_0])\Vdash_{(T_1,\dots,T_n,T_{n+1})} d_3$, i.e., the
sequence\linebreak $(\twins(a_0),T_1,\dots,T_{n+1})$ can be shortened to
$(T_1,\dots,T_{n+1})$.

\begin{lemma}\label{lem:swallow1}
  Let $\autP=(Q,\Delta)$ be a saturated trace-pushdown system. Let
  $T_1,\dots,T_{n+1}\in\twins(\structD)$, $(p_0,a_0,u_0,p_1)\in\Delta$,
  $d_1,d_2,d_3\in\Conf{\autP}$, and $x_0\in A^*$ such that
  \begin{itemize}
    \item $D(a_0)\subseteq D(T_{n+1})$ and
    \item the set $D(T_{n+1})$ is incomparable wrt.\ inclusion with all
      the sets $D(T_i)$ for all $1\le i\le n$.
  \end{itemize}
  Then
  \begin{equation}
    (p_0,[a_0x_0])\vdash(p_1,[u_0x_0])
    \Vdash_{(T_1,T_2,\dots,T_n)} d_1 \vdash^*_\varepsilon d_2
    \vdash_{T_{n+1}} d_3\label{eq:swallow1}
  \end{equation}
  implies
  \begin{equation}
    (p_0,[a_0x_0])\Vdash_{(T_1,\dots,T_n,T_{n+1})} d_3\,.
    \label{eq:swallow2}
  \end{equation}
\end{lemma}

\begin{proof}
  The lemma is shown by induction on the length of the word $u_0$.
  
  If $u_0=\varepsilon$, we have
  \[
    (p_0,[a_0x_0])\vdash_{\varepsilon}
    (p_1,[u_0x_0])
    \vdash_{\varepsilon}^*\,\vdash_{T_1}^+\,
    \Vdash_{(T_2,\dots,T_n)} d_1
    \vdash_\varepsilon^* d_2
    \vdash_{T_{n+1}}d_3
  \]
  and therefore \eqref{eq:swallow2}.
  
  Now suppose $u_0=b u_0'$ with $b\in A$ and $u_0'\in A^*$.  By
  \eqref{eq:swallow1}, there are $k\in\monoN$ and configurations
  $c_j=(q_j,[w_j])$ for all $0\le j\le k$ such that
  \begin{itemize}
    \item $c_0 = (p_0,[a_0x_0])$ and
      $c_1 = (p_1,[u_0x_0])$,
    \item $c_j \vdash c_{j+1}$ for all $1\le j<k$, and
    \item $c_k = d_3$.
  \end{itemize}
  Let $j\in\{1,2,\dots,k-1\}$. Then
  $(q_j,[w_j])=c_j \vdash
  c_{j+1}=(q_{j+1},[w_{j+1}])$ implies the existence of a
  word $x_j\in A^*$ and a transition
  $(q_j,a_j,u_j,q_{j+1})\in\Delta$ such that
  \[
    [w_j]=[a_jx_j] \text{ and }
    [w_{j+1]}= [u_jx_j]
  \]
  (note that this also holds for $j=0$ by the assumption in the
  lemma).
  
  Again by \eqref{eq:swallow1}, the sequence of pairs
  \[
    \bigl((a_j,u_j)\bigr)_{1\le j<k}
  \]
  can be chosen from
  \[
    \left(
    \prod_{1\le m\le n}
    \bigl(A\times\{\varepsilon\}\bigr)^* \bigl(T_m\times A^+\bigr)^+\,
    \right)
    \bigl(A\times\{\varepsilon\}\bigr)^* \bigl(T_{n+1}\times A^+\bigr)\,.
  \]
  In particular, we have $a_{k-1}\in T_{n+1}=\twins(a_{k-1})$, i.e.,
  $D(a_{k-1})=D(T_{n+1})$.
  
  Note that
  \begin{align*}
   b &\in D(u_0) &&(\text{since }u_0\in bA^*)\\
    &\subseteq D(a_0) &&(\text{since }(p_0,a_0,u_0,p_1)\in\Delta)\\
    &\subseteq D(T_{n+1})=D(a_{k-1})\,.
  \end{align*}
  In other words, there is some $\ell\in\{1,2,\dots,k-1\}$ with
  $b\in D(a_\ell)$, we choose $\ell$ minimal with this property implying
  $a_i\parallel b$ for all $1\le i<\ell$.
  
  \begin{claim*}
    $a_\ell=b$.
  \end{claim*}
  
  \begin{claimproof}
    By induction, we construct words $x_i'$ with $x_i\sim b x_i'$ for
    all $1\le i<\ell$. For $i=1$, we have
    \[
      a_1x_1\sim u_0x_0 = b u_0'x_0\,.
    \]
    Since $a_1\parallel b$, this implies the existence of some word
    $x_1'$ with $x_1\sim bx_1'$. Now let $1<i<\ell-1$. Then
    \begin{align*}
      a_ix_i &\sim u_{i-1}x_{i-1}\sim u_{i-1}bx_{i-1}'\\
             &\sim bu_{i-1}x_{i-1}' & (\text{since }b\notin D(a_{i-1})\supseteq D(u_{i-1}))\,.
    \end{align*}
    Since $a_i\parallel b$, this implies the existence of some word
    $x_i'$ with $x_i\sim bx_i'$. Thus, in particular,
    \[
      a_\ell x_\ell\sim u_{\ell-1}x_{\ell-1}\sim u_{\ell-1}bx_{\ell-1}'\sim bu_{\ell-1}x_{\ell-1}'
    \]
    where, as before, the last equivalence follows from
    $b\notin D(a_{\ell-1})\supseteq D(u_{\ell-1})$.  Recall that
    $b\in D(a_\ell)$. Hence
    $a_\ell x_\ell\sim b u_{\ell-1} x_{\ell-1}'$ implies $a_\ell=b$.\qed
  \end{claimproof}
  
  Now, we distinguish two cases, namely whether the word $u_\ell$ is
  empty or not.
  
  First, suppose $u_\ell=\varepsilon$. 
  
  The diamond property of $\autP$ and the independence of $b=a_\ell$
  from all letters $a_1,a_2,\ldots,a_{\ell-1}$ implies that we can
  reorder these transitions. More precisely, there are states
  $p_1,\ldots, p_\ell\in Q$ with $p_\ell=q_{\ell+1}$ such that we obtain
  the following new run from $c_0$ to $c_k$:
  \begin{itemize}
    \item first apply the transition
      $(q_0,a_0,bu_0',q_1)$,
    \item then do the $a_\ell$-transition
      $(q_1, a_\ell, u_\ell, p_1) = (q_1, b, \varepsilon, p_1)$,
    \item then follow the $a_j$-transitions
      $(p_j, a_j, u_j, p_{j+1})$ for $j=1,2,\ldots,\ell-1$,
    \item and finally follow the original path
      $c_{\ell+2},c_{\ell+3},\ldots,c_k$.
  \end{itemize}
  So far, we reordered the path from $c_0$ to $c_k$ in such a way that
  the second transition (starting in the configuration $c_1$) is
  $(q_1, b, \varepsilon, p_1)$.
  
  Note that this reordered path has all the properties of the original
  path $c_0,c_1,c_2,\dots,c_k$, so we abuse notation and call the
  configurations in this new path $c_0,c_1,c_2,\dots,c_k$. The crucial
  effect is that then $\ell=1$ and $u_1=\varepsilon$. In other words,
  we have transitions
  \[
    (q_0,a_0,bu_0',q_1)\text{ and }(q_1,b,\varepsilon,q_2)
  \]
  in the saturated system $\autP$. Saturation implies
  \[
    (q_0,a_0,u_0',q_2)\in\Delta\,.
  \]
  Hence, also the sequence of configurations $c_0,c_2,c_3,\ldots,c_k$
  (omitting $c_1$) is a path in $\autP$ that leads from $c_0$ to
  $c_k$. Note that, in the first transition in this path, $a_0$ is
  replaced by the word $u_0'$ which is properly shorter than
  $u_0$. Hence, by the induction hypothesis, we get
  \eqref{eq:swallow2} which settles the case $u_\ell=\varepsilon$.
  
  Now let $u_\ell\neq\varepsilon$. Then there is $i\in\{1,2,\dots,n+1\}$
  such that $a_\ell\in T_i$, i.e., $D(a_\ell)=D(T_i)$. If $i\le n$, we
  get
  \begin{align*}
    D(T_i) &= D(a_\ell)=D(b) && (\text{since }a_\ell=b)\\
           &\subseteq D(a_0) && (\text{since }(p_0,a_0,bu_0',p_1)=(p_0,a_0,u_0,p_1)\in\Delta)\\
           &\subseteq D(T_{n+1})
  \end{align*}
  which contradicts the incomparability of $D(T_i)$ and $D(T_{n+1})$ for all
  $i\in\{1,\ldots,n\}$. Hence $i=n+1$, i.e.,
  $D(a_\ell)=D(T_{n+1})$. Note that, among all the transitions
  $(q_j,a_j,u_j,q_{j+1})$ for $j\in\{1,2,\dots,k-1\}$, only the last
  one satisfies $D(a_j)=D(T_{n+1})$ and $u_j\neq\varepsilon$. Hence,
  $\ell=k-1$. Thus, we showed $i=n+1$ and $\ell=k-1$.
  
  Note that
  \[
    D(b)\subseteq D(a_0)\subseteq D(T_{n+1})=D(a_\ell)=D(b)
  \]
  implies $D(b)=D(a_0)$.
  Since $a_\ell$ is independent from all the letters from\linebreak
  $\{a_1,a_2,\ldots,a_{\ell-1}\}=\{a_1,a_2,\ldots,a_{k-2}\}$, the same
  applies to $a_0$. Hence, the diamond property of $\autP$ and the
  independence of $a_0$ from all letters $a_1,a_2,\ldots,a_{k-2}$
  imply that we can reorder these transitions. More precisely, there
  are states $p_0,\ldots,p_{k-2}\in Q$ with $p_0=q_0$ such that we have
  the following run from $c_0$ to $c_k$:
  \begin{itemize}
    \item first apply the $a_j$-transitions
      $(p_{j-1}, a_j, u_j, p_j)$ for $j=1,2,\dots,k-2$,
    \item then the $b$-transition
      $(p_{k-2}, a_0, u_0, q_{k-1})$,
    \item and then the $a_{k-1}$-transition
      $(q_{k-1},a_{k-1},u_{k-1},q_k)$.
  \end{itemize}
  Since $D(a_0)= D(b)=D(a_\ell)=D(T_{n+1})$, the resulting run
  satisfies \eqref{eq:swallow2}.
\end{proof}

As explained before the previous lemma, in certain situations, we can
shorten the sequence~$\ov{T}$. To apply the lemma, we require in
particular that, if the run starts with a transition of
$\autP_{\twins(a_0)}$, then this very first phase has length one. Similarly,
the very last phase (a run of $\autP_{T_{n+1}}$) has length one. The
following lemma replaces these requirements by ``non-empty phase'' and
therefore generalizes the above lemma.

\begin{lemma}\label{lem:swallow2}
  Let $\autP=(Q,\Delta)$ be a saturated
  trace-pushdown system. Let $T_0,T_1,\ldots,T_{n+1}\in\twins(\structD)$ such that
  \begin{itemize}
    \item $D(T_0)\subseteq D(T_{n+1})$ and
    \item the sets $D(T_i)$ and $D(T_{n+1})$ are incomparable wrt.~inclusion
      for all $1\le i\le n$.
  \end{itemize}
  Then
  $\Vdash_{(T_0,T_1,\ldots,T_{n+1})} {\subseteq}
  \Vdash_{(T_1,\ldots,T_{n+1})}$.
\end{lemma}

\begin{proof}
  Let $c,d\in\Conf{\autP}$ be two configurations with
  $c\Vdash_{(T_0,T_1,\ldots,T_{n+1})}d$. Then there exist $k\ge1$
  such that
  \begin{equation}
    c\vdash_\varepsilon^* \, 
    \vdash_{T_0}^k \,
    \Vdash_{(T_1,\ldots,T_n)} \,
    \vdash_\varepsilon^* \,
    \vdash_{T_{n+1}}^+ d\,.
    \label{eq:swallow3}
  \end{equation}
  Let $k\ge0$ be minimal such that
  \eqref{eq:swallow3} holds. If $k>0$, there exist
  configurations $c',d'\in\Conf{\autP}$ such that
  \[
    c \vdash_\varepsilon^* \, \vdash_{T_0}^{k-1}
    c' \vdash_{T_0}\,
    \Vdash_{(T_1,\ldots,T_n)}\,
    \vdash_\varepsilon^*\,
    \vdash_{T_{n+1}}
    d' \vdash_{T_{n+1}}^* d
  \]
  From \cref{lem:swallow1}, we obtain
  \[
    c' \Vdash_{(T_1,\ldots,T_n)}
    \,\vdash_\varepsilon^* \, \vdash_{T_{n+1}}^+ d'
  \]
  and therefore
  \[
    c \vdash_\varepsilon^* \, \vdash_{T_0}^{k-1}
    \Vdash_{(T_1,\ldots,T_n)}\,
    \vdash_\varepsilon^*\,
    \vdash_{T_{n+1}}^+ d\,,
  \]
  contradicting our choice of $k$. Hence $k=0$ such that
  \eqref{eq:swallow3} implies $c\Vdash_{(T_1,\ldots,T_{n+1})} d$.
\end{proof}

We will now show that, indeed, the reachability relation of a
saturated system is efficiently lc-rational. The central argument will
be that the above lemma allows to bound the number of phases
uniformly.

\begin{proposition}\label{prop:swallow}
  Let $\structD=(A,D)$ be a dependence alphabet, $\autP=(Q,\Delta)$ a
  trace-pushdown system, and $c,d\in\Conf{\autP}$ two configurations
  of $\autP$ with $c\vdash^*_\autP d$. Then there exist
  $m\le\twinindex(\structD)$ and $\ov{T}=(T_1,\ldots,T_m)$ with
  $T_i\in\twins(\structD)$ such that
  $c\Vdash_{\ov{T}}\,\vdash_\varepsilon^* d$.
\end{proposition}

\begin{proof}
  Since $c\vdash^* d$, there is some sequence of sets
  $\ov{U}=(U_1,U_1,\ldots,U_m)$ with $U_i\in\twins(\structD)$ such that
  \[
    c \Vdash_{\ov{U}}\, \vdash_\varepsilon^* d\,.
  \]
  Let $\ov{U}$ be such a sequence of minimal length and suppose
  $m>\twinindex(\structD)$. Then there are natural numbers $i$ and $n$
  such that
  \begin{equation}
    \label{eq:inclusion-in-sequence}
    1\le i<i+n+1\le m\text{ and }D(U_i)\subseteq D(U_{i+n+1})\,.
  \end{equation}
  Let $n\ge0$ be minimal such that this holds for some $i$.
  
  Then there are configurations $c',d'\in\Conf{\autP}$ such that
  \[
    c \Vdash_{(U_1,\ldots,U_{i-1})}
    c' \Vdash_{(U_i,\ldots, U_{i+n+1})}
    d' \Vdash_{(U_{i+n+2},\ldots,U_m)}\,\vdash_\varepsilon^* d
  \]
  where we understand $\Vdash_{()}$ as identity relation. For notational
  simplicity, we set $T_j=U_{i+j}$ for all $0\le j\le n+1$ such that
  \[
    c' \Vdash_{(T_0,T_1,\ldots,T_{n+1})} d'\text{ and }
    D(T_0)\subseteq D(T_{n+1})\,.
  \]
  Let $j\in\{1,2,\ldots,n\}$. Then $D(T_j)\nsubseteq D(T_{n+1})$ since
  otherwise we would have chosen $n-j<n$ for $n$. Similarly,
  $D(T_{n+1})\nsubseteq D(T_j)$ since otherwise
  $D(T_0)\subseteq D(T_{n+1})\subseteq D(T_j)$ and we would have chosen $j-i<n$
  for $n$. Hence we showed that the sets $D(T_j)$ and $D(T_{n+1})$ are
  incomparable for all $j\in\{1,2,\dots,n\}$. From \cref{lem:swallow2},
  we obtain
  \[
    c' \Vdash_{(T_1,\ldots,T_{n+1})} d'\,.
  \]
  But this implies
  \[
    c \Vdash_{(U_1,\ldots,U_{i-1},U_{i+1},\ldots,U_m)} \vdash_\varepsilon^* d
  \]
  contradicting our choice of $m$.

  Hence, indeed, $m\le\twinindex(\structD)$.
\end{proof}

\begin{theorem}\label{thm:reachSaturated}
  Let $\structD=(A,D)$ be a dependence alphabet, $\autP=(Q,\Delta)$ a
  saturated trace-pushdown system over $\structD$, and $p,q\in
  Q$. Then the reachability relation $\reach_{p,q}(\autP)$ is
  lc-rational.

  More precisely, from $\structD$, $\autP$, $p$, and $q$, one can
  compute a left-closed transducer $\autT$ with
  $\reach_{p,q}(\autP)=[R(\autT)]$; this computation can be carried
  out in time $\|\autP\|^{\mathrm{poly}(\twinindex(\structD))}$.
\end{theorem}

\begin{proof}
  In this proof, let $n=2\cdot\twinindex(\structD)+1$.

  For states $r_1,r_2\in Q$, let
  \[
    \calR_{r_1,r_2}(\autP)=\reach_{r_1,r_2}(\autP_\varepsilon)\cup
    \bigcup_{T\in\twins(\structD)}\reach_{r_1,r_2}(\autP_T)
  \]
  denote the union of the reachability relations of the homogeneous
  subsystems $\autP_\varepsilon$ and $\autP_T$ of $\autP$.  Recall
  that $\Delta_T\neq\emptyset$ implies the existence of some $a\in A$
  with $\twins(a)=T$, hence
  \[
    \calR_{r_1,r_2}(\autP)=\reach_{r_1,r_2}(\autP_\varepsilon)\cup
    \bigcup_{b\in A}\reach_{r_1,r_2}(\autP_{\twins(b)})\,.
  \]
  This relation is lc-rational according to
  \cref{prop:read,prop:write}. Now Proposition~\ref{prop:swallow} implies
  \[
    \reach_{p,q}(\autP)=\bigcup_{\substack{q_0,q_1,\ldots,q_n\in Q,\\
        q_0=p,q_n=q}} \prod_{0\leq i<n}\calR_{q_i,q_{i+1}}(\autP)
  \]
  for any $p,q\in Q$ (here, $\prod$ represents the iterated
  application of $\circ$). Since $\calR_{q_i,q_{i+1}}(\autP)$ is
  lc-rational for any pair $0\le i<n$ and since the class of
  lc-rational trace relations is closed under composition
  (\cref{prop:composition}) and union, the trace relation
  $\reach_{p,q}(\autP)$ is also lc-rational.

  It remains to verify the claimed upper bound.

  First, fix some sequence $\ov{q}=(q_0,\ldots,q_n)$ of states of
  $\autP$. Let $0\le i<n$. By \cref{prop:read,prop:write}, we can
  construct, for all $X\in\{\varepsilon\}\cup\twins(\structD)$, a
  left-closed transducer $\autT_{q_i,q_{i+1},X}$ with
  $[R(\autT_{q_i,q_{i+1},X})]=\reach_{q_i,q_{i+1}}(\autP_X)$ in time
  polynomial in $\|\autP\|+\settwinindex(\structD)$. Let the left-closed
  transducer $\autT_{q_i,q_{i+1}}$ be the disjoint union of all the
  transducers $\autT_{q_i,q_{i+1},Z}$ such that
  $[R(\autT_{q_i,q_{i+1}})]=\calR_{q_i,q_{i+1}}$. Since the number of
  possible values for $X$ is polynomial in $\twinindex(\structD)$, the
  left-closed transducer $\autT_{q_i,q_{i+1}}$ can be constructed in
  time polynomial in $\|\autP\|+\settwinindex(\structD)+\twinindex(\structD)$.

  Let $t$ denote the total size of all these transducers
  $\autT_{q_i,q_{i+1}}$ as well as the tPDS $\autP$ (since the number
  $n$ of these transducers is polynomial in $\twinindex(\structD)$,
  the number $t$ is polynomial in
  $\|\autP\|+\settwinindex(\structD)+\twinindex(\structD)$). By
  \cref{prop:composition}, a left-closed transducer
  $\autT_{\overline{q}}$ with
  $[R(\autT_{\overline{q}})]=\prod_{0\le i\le
    n}[R(\autT_{q_i,q_{i+1}})]$ can be computed in time
  $O(t^n)$.

  Since there are $|Q|^{n+1}\le t^{n+1}$ many such sequences $\ov{q}$
  of states, a left-closed transducer $\autT$ for the union of all
  these relations can be obtained in time
  $t^{\mathrm{poly}(n)}=t^{\mathrm{poly}(\twinindex(\structD))}$. Recall
  that $t$ is polynomial in
  $\|\autP\|+\settwinindex(\structD)+\twinindex(\structD)$. The number of sets
  $D(B)$ for $B\subseteq A$ is at most $2^{\twinindex(\structD)}$ since any set
  $D(B)$ is the union of the sets $D(b)$ for $b\in B$. Hence
  $\settwinindex(\structD)\le 2^{\twinindex(\structD)}$ implying that $t$ is
  polynomial in $\|\autP\|^{\mathrm{poly}(\twinindex(\structD))}$. But
  this implies that the construction of $\autT$ can, indeed, be
  carried out in time
  $\|\autP\|^{\mathrm{poly}(\twinindex(\structD))}$.
\end{proof}

\subsection{Saturating a system}
\label{ssec:step1}

So far, we showed that the reachability relation is lc-rational,
provided the system is saturated. To get the result in full
generality, it remains to transform an arbitrary system into an
equivalent saturated one. For a classical pushdown system (i.e., a
trace-pushdown system over the trace monoid $A^*$), the idea is very
simple: if there are transitions $(p,a,bw,q)$ and
$(q,b,\varepsilon,r)$, then adding the transition $(p,a,w,r)$ does not
change the behavior and transforms the system closer to a saturated
one. In the trace-pushdown setting, the technicalities are a bit more
involved: suppose we have the transitions $(p,a,cbw,q)$ and
$(q,b,\varepsilon,r)$ with $b\parallel c$. Then $[cbw]=[bcw]$, i.e.,
after doing the first transition (that writes the trace
$[cbw]=[bcw]$ onto the pushdown), the second transition (eliminating $b$) can be
executed immediately. Therefore, also in this situation, we add the
transition $(p,a,cw,r)$ to get closer to a saturated system.

We also want to keep the system ``small''. In particular, we do not
want to add transitions $(p,a,u,q)$ and $(p,a,v,q)$ with $u\sim v$ as
they are redundant. To achieve this technically, we use lexicographic
normal forms defined as follows.

\begin{definition}
  Let $\structD=(A,D)$ be a dependence alphabet and $\leq$ a linear
  order on $A$. Further, let $u\in A^*$. Then $\lnf(u)$ is the (wrt.\
  $\leq$) lexicographically minimal word $v\in A^*$ with $u\sim v$; it
  is called the \emph{lexicographic normal form} of $u$.
\end{definition}

Note that, since only finitely many words $v$ satisfy $u\sim v$, the
lexicographic normal form exists for any word $u$. Further, it can
easily be computed from $u$ in polynomial time
(cf.~\cite[Section~1.5]{Die90}).

Formally, we construct the trace-pushdown systems
$\autP^{(k)}=(Q,\Delta^{(k)})$ for any $k\in\monoN$ as
follows:
\begin{itemize}
  \item we set $\Delta^{(0)}:=\{(p,a,\lnf(w),q)\mid (p,a,w,q)\in\Delta\}$.
  \item To obtain $\Delta^{(k+1)}$, we add to the set $\Delta^{(k)}$ all
    transitions $(p,a,\lnf(uv),r)$ for which there are a letter
    $b\in A$ and a state $q\in Q$ such that
    $(p,a,ubv,q),$ $(q,b,\varepsilon,r)\in\Delta^{(k)}$
    and $u\parallel b$.
\end{itemize}
Let $\Delta^{(\infty)}=\bigcup_{k\ge0}\Delta^{(k)}$ be the ``limit''
of the increasing sequence of sets $\Delta^{(k)}$. Since $\Delta$ is
finite, there exists a natural number $N$ with
$\Delta\subseteq Q \times A \times A^{\le N} \times Q$, i.e., all
words written by the trace-pushdown system $\autP$ are of length at
most $N$. The construction ensures that this also holds for all
pushdown systems $\autP^{(k)}=(Q,\Delta^{(k)}$. Hence the limit set
$\Delta^{(\infty)}$ is finite, i.e., the pair
$\autP^{(\infty)}=(Q,\Delta^{(\infty)})$ is a pushdown system. Since
the sequence of sets $\Delta^{(k)}$ is increasing, there is some
natural number $\ell\le\bigl|Q\times A\times A^{\le N}\times Q\bigr|$
with $\Delta^{(\infty)}=\Delta^{(\ell)}$ (later, we will see that a
much smaller number $\ell$ suffices).

\begin{example}
  In \cref{fig:apds-shortcuts} we depict our construction of the
  pushdown systems $\autP^{(k)}$ for $k\in\{0,1,2\}$ starting with the
  tPDS $\autP=\autP^{(0)}$ on the left of
  \cref{fig:apds-shortcuts}. It can be verified that
  $\autP^{(2)}=\autP^{(3)}$. Further, all the pushdown systems
  $\autP^{(k)}$ satisfy conditions~\ref{def:tPDS1} and
  \ref{def:diamond1}.
  
  Below, we prove these two properties for all trace-pushdown systems
  $\autP$.
\end{example}

\begin{figure}[h]
  \begin{center}
    \begin{tikzpicture}[nfa,scale=0.85,every node/.append style={scale=0.85}]
      \node[elliptic state] (p00) at (1.125,1.9486) {$0$};
      \node[elliptic state] (p10) {$1$};
      \node[elliptic state,right of=p10] (p20) {$2$};
      \draw[-,dashed,ultra thick] (3.375,2.45) -- (3.375,-1.35);
      \node[right of=p00] (p00h) {};

      \node[elliptic state,right of=p00h] (p01) {$0$};
      \node[elliptic state,right of=p20] (p11) {$1$};
      \node[elliptic state,right of=p11] (p21) {$2$};
      \draw[-,dashed,ultra thick] (7.875,2.45) -- (7.875,-1.35);
      \node[right of=p01] (p01h) {};

      \node[elliptic state,right of=p01h] (p02) {$0$};
      \node[elliptic state,right of=p21] (p12) {$1$};
      \node[elliptic state,right of=p12] (p22) {$2$};

      \draw (p00) edge node[left] {$a\mid ab$} (p10)
            (p10) edge[loop below] node {$a\mid ab,b\mid\varepsilon$} (p10)
                  edge node {$a\mid\varepsilon,c\mid\varepsilon$} (p20)
            (p20) edge[loop below] node {$b\mid\varepsilon,c\mid\varepsilon$} (p20)
            (p01) edge node[left] {$a\mid ab$} (p11)
                  edge[color=darkred,thick] node {$\boldsymbol{a\mid b}$} (p21)
            (p11) edge[loop below] node {$a\mid ab,b\mid\varepsilon$} (p101)
                  edge node {$a\mid\varepsilon,c\mid\varepsilon$} node[below] {$\boldsymbol{\color{darkred}a\mid b}$} (p21)
            (p21) edge[loop below] node {$b\mid\varepsilon,c\mid\varepsilon$} (p21)
            (p02) edge node[left] {$a\mid ab$} (p12)
                  edge node {$a\mid b,\boldsymbol{\color{darkred}a\mid\varepsilon}$} (p22)
            (p12) edge[loop below] node {$a\mid ab,b\mid\varepsilon$} (p12)
                  edge node {$a\mid\varepsilon,c\mid\varepsilon$} node[below] {$a\mid b$} (p22)
            (p22) edge[loop below] node {$b\mid\varepsilon,c\mid\varepsilon$} (p22);
    \end{tikzpicture}
  \end{center}
  \caption{The trace-pushdown system $\autP=\autP^{(0)}$,
    $\autP^{(1)}$, and $\autP^{(2)}=\autP^{(\infty)}$ (from left to
    right). New transitions are marked in bold and
    red.\label{fig:apds-shortcuts}}
\end{figure}
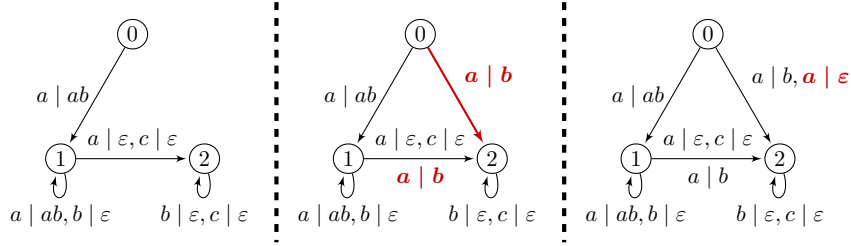

By definition, the pair $\autP^{(k)}=(Q,\Delta^{(k)})$ is a pushdown
system. The following lemma states that this pushdown system is even a
trace-pushdown system.

\begin{lemma}\label{lem:saturationCorrect1}
  The pair $\autP^{(\infty)}$ is a trace-pushdown system.
\end{lemma}

\begin{proof}
  Since $\autP^{(\infty)}$ equals one of the systems $\autP^{(k)}$, it
  suffices to prove this statement by induction on $k$ for all systems
  $\autP^{(k)}$. The base case $k=0$ is obvious since $\autP$ is a
  trace-pushdown system.

  Now let $k\ge0$.
  We have to verify that $\autP^{(k+1)}$ satisfies the two conditions
  \ref{def:tPDS1} and \ref{def:diamond1} from \cref{def:cpds} and
  \cref{lem:(P2)-suffices}, respectively.

  To show \ref{def:tPDS1}, let $(p,a,w,r)\in\Delta^{(k+1)}$. If this
  transition belongs to $\Delta^{(k)}$, we obtain $D(w)\subseteq D(a)$
  by the induction hypothesis. So suppose, in the other case,
  $(p,a,w,r)\in\Delta^{(k+1)}\setminus\Delta^{(k)}$.  Then, by
  definition of $\Delta^{(k+1)}$, there are words $u,v\in A^*$ with
  $w_1=\lnf(uv)$, a letter $b\in A$ with $u\parallel b$, and a state
  $q\in Q$ such that
  $(p,a,ubv,q),(q,b,\varepsilon,r)\in\Delta^{(k)}$. Then we obtain
  \[
    D(w)=D(uv)\subseteq D(ubv)\subseteq D(a)
  \]
  where the last inclusion follows from the induction hypothesis since
  $(p,a,ubv,q)\in\Delta^{(k)}$. Hence $\autP^{(k+1)}$ satisfies
  property~\ref{def:tPDS1}.

  We next show \ref{def:diamond1}. So let
  $(p,a,w,q),(q,c,x,r)\in\Delta^{(k+1)}$ with $a\parallel c$.

  We distinguish the cases $(q,c,x,r)\in\Delta^{(k)}$ and
  $(q,c,x,r)\in\Delta^{(k+1)}\setminus\Delta^{(k)}$.

  \noindent
  \textbf{Case 1.} Suppose $(q,c,x,r)\in\Delta^{(k)}$.

  If, in addition, $(p,a,w,q)\in\Delta^{(k)}$, then the induction
  hypothesis yields a state $q'$ with
  $(p,c,x,q'),(q',a,w,r)\in\Delta^{(k)}\subseteq\Delta^{(k+1)}$.

  So suppose $(p,a,w,q)\in\Delta^{(k+1)}\setminus\Delta^{(k)}$ (cf.\
  \cref{fig:saturationCorrect:f1} for a visualization of the proof
  where all arrows belong to $\Delta^{(k+1)}$ and the dashed ones even
  to $\Delta^{(k)}$ ). By the definition of $\Delta^{(k+1)}$, there
  are words $u,v\in A^*$ with $w=\lnf(uv)$, a letter $b\in A$ with
  $u\parallel b$, and a state $s\in Q$ such that
  $(p,a,ubv,s),(s,b,\varepsilon,q)\in\Delta^{(k)}$.

  Note that $(s,b,\varepsilon,q),(q,c,x,r)\in\Delta^{(k)}$. We know
  $D(b)\subseteq D(ubv)\subseteq D(a)$ where the last inclusion
  follows from $(p,a,ubv,s)\in\Delta^{(k)}$ since $\autP^{(k)}$
  satisfies \ref{def:tPDS1}. From $a\parallel c$, we obtain
  $c\notin D(a)\supseteq D(b)$, i.e., $b\parallel c$.  Since the tPDS
  $\autP^{(k)}$ satisfies \ref{def:diamond1}, we obtain a state
  $q'\in Q$ with transitions
  $(s,c,x,q'),(q',b,\varepsilon,r)\in\Delta^{(k)}$.

  Note further that $(p,a,ubv,s),(s,c,x,q')\in\Delta^{(k)}$ and
  $a\parallel c$. Hence the diamond property \ref{def:diamond1} in
  $\autP^{(k)}$ implies the existence of a state $s'\in Q$ with
  $(p,c,x,s'),(s',a,ubv,q')\in\Delta^{(k)}$.

  Finally note that
  $(s',a,ubv,q'),(q',b,\varepsilon,r)\in\Delta^{(k)}$. Since
  $b\parallel u$ and $w=\lnf(uv)$, the construction ensures
  $(s',a,w,r)\in\Delta^{(k+1)}$.

  In summary, we found a state $s'\in Q$ with
  $(p,c,x,s')\in\Delta^{(k)}$ and $(s',a,w,r)\in\Delta^{(k+1)}$. Since
  $\Delta^{(k)}\subseteq\Delta^{(k+1)}$, this finishes the
  verification of \ref{def:diamond1} in case 1.

  \noindent
  \textbf{Case 2.} Suppose $(p,a,w,q)\in\Delta^{(k+1)}$ and
  $(q,c,x,r)\in\Delta^{(k+1)}\setminus\Delta^{(k)}$ (cf.\
  \cref{fig:saturationCorrect:f2} for a visualization of the proof).

  By construction of $\Delta^{(k+1)}$, there are words
  $y,z\in A^*$ with $x=\lnf(yz)$, a letter $d\in A$ with
  $y\parallel d$, and a state $s\in Q$ such that
  $(q,c,ydz,s),(s,d,\varepsilon,r)\in\Delta^{(k)}$. 

  Note that $(p,a,w,q)\in\Delta^{(k+1)}$ and
  $(q,c,ydz,s)\in\Delta^{(k)}$ with $a\parallel c$. Hence, by case 1,
  there is a state $q'\in Q$ with $(p,c,xdy,q')\in\Delta^{(k)}$ and
  $(q',a,w,s)\in\Delta^{(k+1)}$.

  Note further that $(q',a,w,s)\in\Delta^{(k+1)}$ and
  $(s,d,\varepsilon,r)\in\Delta^{(k)}$. We have
  $D(d)\subseteq D(ydz)\subseteq D(c)$ where the last inclusion
  follows from $(q,c,ydz,s)\in\Delta^{(k)}$ since $\autP^{(k)}$
  satisfies \ref{def:tPDS1} by the induction hypothesis. From
  $a\parallel c$, we obtain $a\notin D(c)\supseteq D(d)$ and therefore
  $a\parallel d$. Hence, again by case 1, there is a state $s'$ with
  $(q',d,\varepsilon,s')\in\Delta^{(k)}$ and
  $(s',a,w,r)\in\Delta^{(k+1)}$.

  Finally note that
  $(p,c,ydz,q'),(q',d,\varepsilon,s')\in\Delta^{(k)}$. Since
  $d\parallel y$ and $\lnf(x)=yz$, the construction of $\Delta^{(k+1)}$ yields
  $(p,c,x,s')\in\Delta^{(k+1)}$.

  Thus, we found a state $s'\in Q$ with
  $(p,c,x,s'),(s',a,w,r)\in\Delta^{(k+1)}$ which finishes the
  verification of \ref{def:diamond1} in case 2.

  As the above two cases cover all possibilities, we showed that
  $\autP^{(k+1)}$ satisfies \ref{def:diamond1}.  \qedhere
  
  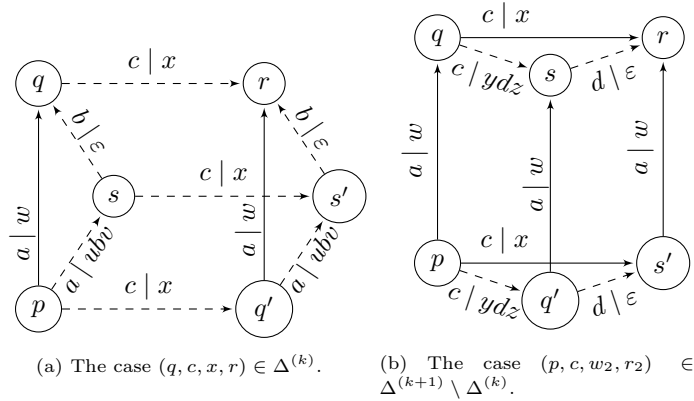
\begin{figure}[h]
    \centering
    \subcaptionbox{The case $(q,c,x,r)\in\Delta^{(k)}$.\label{fig:saturationCorrect:f1}}{
      \begin{tikzpicture}[nfa]
        \node[state] (p) at (0,0) {$p$};
        \node[state] (q1) at (0,3) {$q$};
        \node[state] (r) at (3,3) {$r$};
        \node[state] (s1) at (1,1.5) {$s$};
        \node[state] (q2) at (4,1.5) {$s'$};
        \node[state] (s2) at (3,0) {$q'$};
        
        \draw (p) edge node[rotate=90, left=6pt] {$a\mid w$} (q1)
                  edge[dashed] node[sloped,below=-3pt] {$a\mid ubv$} (s1)
                  edge[dashed] node {$c\mid x$} (s2)
              (s1) edge[dashed] node[sloped,above=-3pt] {$b\mid \varepsilon$} (q1)
                   edge[dashed] node {$c\mid x$} (q2)
              (s2) edge[dashed] node[sloped,below=-3pt] {$a\mid ubv$} (q2)
                   edge node[rotate=90, left=6pt] {$a\mid w$} (r)
              (q1) edge[dashed] node {$c\mid x$} (r)
              (q2) edge[dashed] node[sloped,above=-3pt] {$b\mid\varepsilon$} (r);
      \end{tikzpicture}
    }
    \subcaptionbox{The case $(p,c,w_2,r_2)\in\Delta^{(k+1)}\setminus\Delta^{(k)}$.\label{fig:saturationCorrect:f2}}{
      \begin{tikzpicture}[nfa]
        \node[state] (p) at (0,1) {$p$};
        \node[state] (q) at (0,4) {$q$};
        \node[state] (r) at (3,4) {$r$};
        \node[state] (s) at (1.5,3.5) {$s$};
        \node[state] (q') at (1.5,0.5) {$q'$};
        \node[state] (s') at (3,1) {$s'$};

        \draw (p) edge[sloped,above=-3pt] node {$a\mid w$} (q)
                  edge node[near start] {$c\mid x$} (s')
                  edge[dashed] node[sloped,below] {$c\mid ydz$} (q')
              (q) edge node[near start] {$c\mid x$} (r)
                  edge[dashed] node[sloped,below] {$c\mid ydz$} (s)
              (q') edge node[sloped,above=-3pt] {$a\mid w$} (s)
                   edge[dashed] node[sloped,below] {$d\mid\varepsilon$} (s')
              (s) edge[dashed] node[sloped,below] {$d\mid\varepsilon$} (r)
              (s') edge[sloped,above=-3pt] node {$a\mid w$} (r);
      \end{tikzpicture}
    }
    \caption{Visualization of the proof of the diamond properties
      (cf.~proof of Lemma~\ref{lem:saturationCorrect1}). All
      arrows vizualize transitions from $\Delta^{(k+1)}$ and dashed
      arrows transitions that even belong to
      $\Delta^{(k)}$).\label{fig:saturationCorrect}}
  \end{figure}
\end{proof}

%%%%%%%%%%%%%%%%%%%%%%%%%%%%%%%%%%%%%%%%%%%%%%%%%%%%%%%%%%%%%%%%%%%%%% 

Since $\autP$ and $\autP^{(\infty)}$ have the same alphabet and the same
set of states, the sets of configurations of these two trace-pushdown
systems coincide. The following lemma states that they also agree in
their reachability relations.

\begin{lemma}\label{lem:saturationCorrect2}
  For any two configurations $(p,[u])$ and $(q,[v])$, we have
  $(p,[u]) \vdash^*_\autP (q,[v])$ if, and only if,
  $(p,[u]) \vdash^*_{\autP^{(\infty)}} (q,[v])$.
\end{lemma}

\begin{proof}
  We first demonstrate the implication ``$\Rightarrow$''.
  For this, it suffices to show that
  $(p,[u])\vdash_{\autP}(q,[v])$ implies
  $(p,[u])\vdash_{\autP^{(\infty)}}(q,[v])$. But
  this is trivial since the transitions $(p,a,w,q)\in\Delta$
  and $(p,a,\lnf(w),q)\in\Delta^{(0)}\subseteq\Delta^{(\infty)}$
  have the same effect on any configuration. For the other
  implication, it suffices to show for any natural number $k$ that
  $(p,[u]) \vdash_{\autP^{(k+1)}} (r,[v])$
  implies
  $(p,[u]) \vdash_{\autP^{(k)}}^* (r,[v])$,
  i.e., that every single step of $\autP^{(k+1)}$ can be simulated by
  a sequence of steps of $\autP^{(k)}$.
  
  Since $(p,[u]) \vdash_{\autP^{(k+1)}} (r,[v])$,
  there is a transition $(p,a,w,r)\in\Delta^{(k+1)}$ and a
  word $u'$ such that $u\sim au'$ and
  $v\sim wu'$. If the transition $(p,a,w,r)$
  belongs to $\Delta^{(k)}$, we immediately get
  $(p,[u]) \vdash_{\autP^{(k)}}
  (r,[v])$. Otherwise, there are words $w_1$ and $w_2$ with
  $w=\lnf(w_1w_2)$ such that, for some letter $b\in A$ and some
  state $q\in Q$, we have
  $(p,a,w_1bw_2,q),(q,b,\varepsilon,r)\in\Delta^{(k)}$
  and $w_1\parallel b$. But then $w_1bw_2\sim bw_1w_2$
  implying
  \begin{align*}
    (p,[u])=(p,[au']) &\vdash_{\autP^{(k)}} (q,[w_1bw_2u'])
    = (q,[bw_1w_2u'])\\
    &\vdash_{\autP^{(k)}} (r,[w_1w_2u'])
    = (r,[v])\,.
  \end{align*}
  Hence, also the implication ``$\Leftarrow$'' of the claim holds.
\end{proof}

So far, we proved that the reachability relation of all saturated
trace-push\-down systems are lc-rational and that all trace-pushdown
system can be saturated, i.e., the reachability relation of all
trace-pushdown systems are lc-rational. Our main result
\cref{thm:reachLCRational} also claims that a transducer for this
relation can be constructed in polynomial time (assuming the
dependence alphabet $\structD$ to be fixed). Since
\cref{thm:reachSaturated} contains a corresponding statement, it
remains to be shown that the above saturation procedure can be
carried out in polynomial time (and therefore leads to a system
of polynomial size).

Recall that the independence number $\alpha(\structD)$ is the maximal
size of a set of letters that are mutually independet.

\begin{lemma}
  Let $\structD=(A,D)$ be a dependence alphabet and $\autP$ be a
  trace-pushdown system.  The number of transitions in
  $\Delta^{(\infty)}$ is polynomial in $\|\autP\|^{\alpha(\structD)}$
  and this set can be computed in time polynomial in
  $\|\autP\|^{\alpha(\structD)}$.
\end{lemma}

\begin{proof}
  For $k\ge0$, let
  $W_k=\{w\in A^*\mid \exists(p,a,w,r)\in\Delta^{(k)}\}$
  denote the set of words written by some transition from $\Delta^{(k)}$.
  
  By induction on $k$, we show that, for any word $w\in W_k$, there
  exists a word $v\in W_0$ such that the trace $[w]$ is a
  \emph{suffix} of the trace $[v]$ (i.e., $v\sim uw$ for some word
  $u\in A^*$).  Setting $v:=w$, this is trivial for $k=0$. Now let
  $k\ge0$ and $w\in W_{k+1}\setminus W_k$. Then there is a transition
  $(p,a,w,r)\in\Delta^{(k+1)}\setminus \Delta^{(k)}$ that is added in
  step $k+1$ of our construction. Hence there is some transition
  $(p,a,ubv,q)\in\Delta^{(k)}$ with $u\parallel b$ and
  $w=\lnf(uv)$. Hence
  \begin{align*}
    ubv &\sim  b\,uv && (\text{since }u\parallel b)\\
    &\sim  b\,w  && (\text{since }w=\lnf(uv))
  \end{align*}
  implying that $[w]$ is a suffix of $[ubv]$ and $ubv\in W_k$. By the
  induction hypothesis, there is some $x\in W_0$ such that $[ubv]$
  is a suffix of $[x]$. Now the transitivity of the suffix
  relation implies the claim for $w$ which completes the inductive
  proof.
  
  Let $N$ be the maximal length of a word from $W_0$. From
  \cite{BerMS89}, we obtain that the set of \textit{traces}
  \[
    W_k'=\{[w]\in A^*\mid \exists(p,a,w,r)\in\Delta^{(k)}\}
  \]
  contains at most $|W_0|\cdot N^{\alpha(\structD)}$ elements since
  all traces from $W_k'$ are suffixes of traces from $W_0'$. Since,
  for any word $w\in A^*$, there is a unique word $v$ in lexicographic
  normal form with $w\sim v$, we get
  \[
    |W_k|=|W_k'|\le |W_0|\cdot N^{\alpha(\structD)}\,.
  \]
  It follows that
  \[
    |\Delta^{(k)}|\le\left| Q\times A\times Q\right|\cdot |\Delta|\cdot N^{\alpha(\structD)}
  \]
  which is polynomial in $\|\autP\|^{\alpha(\structD)}$.
\end{proof}

The main properties of the above saturation procedure can be
summarized as follows:

\begin{proposition}\label{prop:saturation}
  Let $\structD=(A,D)$ be a dependence alphabet and $\autP$ a
  trace-pushdown system. Then, in time polynomial in
  $\|\autP\|^{\alpha(\structD)}$, one can construct an equivalent
  saturated system $\autP^{(\infty)}$, i.e., a saturated
  trace-pushdown system with
  $\vdash_\autP^* {=} \vdash_{\autP^{(\infty)}}^*$.\qed
\end{proposition}

\subsection{Proof of \cref{thm:reachLCRational}}
\label{sec:proof-forwards}

We finally summarize our proof of the fact that the reachability
relation of trace-pushdown systems is efficiently lc-rational.

So let $\structD=(A,D)$ be a dependence alphabet and $\autP$ a
trace-pushdown system.
\begin{enumerate}
\item Using \cref{prop:saturation}, we can construct an equivalent
  saturated system $\autP^{(\infty)}$ in time polynomial in
  $\|\autP\|^{\alpha(\structD)}$. Since the independence number
  $\alpha(\structD)$ is bounded by the twin index
  $\twinindex(\structD)$, the construction of $\autP^{(\infty)}$ is
  possible in time polynomial in $\|\autP\|^{\twinindex(\structD)}$
  \item Now the claim follows from \cref{thm:reachSaturated}.\qed
\end{enumerate}

\section{Closure Properties of the Reachability Sets of tPDS}

From \cite{Cau88}, we know that classical pushdown systems efficiently
preserve the regularity of sets of configurations under forwards and
under backwards reachability. More formally, suppose $\autP$ is a
pushdown system (i.e., a trace-pushdown system over the trace monoid
$A^*$), $p,q,r$ are three states, and $\autB$ is an NFA over the
alphabet $A$. Let $C=\{q\}\times L(\autB)$ denote the set of
configurations with state $q$ whose pushdown content belongs to the
language $L(\autB)$. Then $\{u\mid (p,u)\vdash^*_\autP C\}$ is the set
of pushdown contents that allow to reach some pushdown content from
$L(\autB)$ (while changing state from $p$ to $q$). Caucal's result
says that this language is efficiently regular, i.e., that we can
compute in polynomial time an NFA $\autA$ accepting this
set. Symmetrically, $\{w\mid C\vdash^*_\autP(r,w)\}$ is the set of
pushdown contents that can be reached from some pushdown content from
$L(\autB)$ (while changing state from $q$ to $r$). Caucal's result
implies that also this language is efficiently regular, i.e., that we can
compute in polynomial time an NFA $\autC$ accepting this set.

In this section, we ask to what extent these two results generalize to
trace-pushdown systems. Since now, we are concerned with sets of
traces as opposed to sets of words, we can consider rational or
recognizable such sets. Accordingly, we ask whether the rationality or
the recognizability is preserved under forwards and backwards
reachability in a trace-pushdown system.

In \cite{KoeK24}, we consider cooperating multi-pushdown sytems. Any
such system is a trace-pushdown system, but not conversely.  There, we
present a cooperating multi-pushdown system and a rational set $C$ of
configurations such that the set of configurations backwards reachable
from $C$ is not even decidable, let alone rational,
cf.~\cite[Prop.~17]{KoeK24}.

In \cite{KoeK24}, we also show that the backwards reachability of
cooperating multi-pushdown systems preserves the recognizability of a
set of configurations. Next, we generalize this result to tPDS and we
prove, in addition, that tPDS preserve the rationality under forwards
reachability.

\begin{theorem}\label{thm:reach}
  Let $\structD=(A,D)$ be a dependence alphabet and $\autP=(Q,\Delta)$
  a trace-pushdown system. Then the following two statements hold:
  \begin{enumerate}[label=(\arabic*)]
  \item The backwards reachability relation preserves
    recognizability. More precisely, from $\structD$, $\autP$,
    $p,q\in Q$, and from a \emph{closed NFA} $\autB$, one can compute a
    \emph{closed NFA} $\autA$ such that
    \[
      [L(\autA)]=\{[u]\mid (p,[u])\vdash^*_\autP\{q\}\times[L(\autB)]\}\,;
    \]
    this computation can be carried out in time polynomial in
    $\|\autP\|^{\mathrm{poly}(\twinindex(\structD))} + \|\autB\|$.
    \label{thm:reach:i1}
  \item The forwards reachability relation preserves rationality. More
    precisely, from $\structD$, $\autP$, $q,r\in Q$, and from an
    \emph{NFA} $\autB$, one can compute an \emph{NFA} $\autC$ such
    that
    \[
      [L(\autC)]=\{[w]\mid \{q\}\times[L(\autB)]\vdash^*_\autP(r,[w])\}\,;
    \]
    this computation can be carried out in time polynomial in
    $\|\autP\|^{\mathrm{poly}(\twinindex(\structD))} + \|\autB\|$.
    \label{thm:reach:i2}
  \end{enumerate}
\end{theorem}

\begin{proof}
  We first prove statement \ref{thm:reach:i1}. Let
  $\calL=\{[u]\mid (p,[u])\vdash^*_\autP\{q\}\times[L(\autB)]\}$ such
  that we have to construct a closed NFA $\autA$ with
  $[L(\autA)]=\calL$.
  
  Recall that the relation $\calR:=\reach_{p,q}(\autP)$ is the set of
  pairs $([u],[v])$ such that
  $(p,[u])\vdash^*_\autP(q,[v])$. Consequently,
  $\calL = {}^\calR[L(\autB)]$.

  By \cref{thm:reachLCRational}, the relation $\calR$ is
  lc-rational. Since the NFA $\autB$ is closed,
  \cref{thm:preservation}(i) implies that the set $\calL$ is
  recognizable, i.e., that there exists a closed NFA $\autA$ with
  $[L(\autA)]=\calL$.

  To verify the upper time bound, recall that a left-closed transducer
  $\autT$ for $\calR$ can be computed in time polynomial in
  $\|\autP\|^{\mathrm{poly}(\twinindex(\structD))}$
  by \cref{thm:reachLCRational}. From this left-closed transducer
  $\autT$ and the closed NFA $\autB$, we can construct the closed NFA
  $\autA$ in time polynomial in
  $\|\structD\|+\|\autT\|+\|\autB\|\le\|\autT\|+\|\autB\|$. Consequently,
  $\autA$ can be constructed in time polynomial in
  $\|\autP\|^{\mathrm{poly}(\twinindex(\structD))}
  + \|\autB\|$.

  Statement \ref{thm:reach:i2} can be shown similarly, using
  \cref{thm:preservation}(ii) instead of \cref{thm:preservation}(i).
\end{proof}

It remains to consider the question whether a tPDS preserves the
recognizability under forwards reachability. The answer to this
question is negative: every rational trace language $\calL$ is the set
of configurations reachable from some single configuration. Since $\calL$
can be non-recognizable, this implies that the recognizability is not
preserved under forwards reachability:

\begin{proposition}\label{prop:rational}
  Let $\structD=(A,D)$ be a dependence alphabet and
  $\calL\subseteq\traceMon$ be rational. Then there are a dependence
  alphabet $\structD'=(A',D')$ with $A\subseteq A'$ and
  $D = D'\cap (A\times A)$, a trace-pushdown system
  $\autP=(Q,\Delta)$ over $\structD'$,
  $c\in\Conf{\autP}$, and $q\in Q$ such that
  $\calL=\{[w]\mid c\vdash^*(q,[w])\}$.
\end{proposition}

\begin{proof}
  Since $\calL$ is rational, there is an NFA $\autA=(S,A,I,T,F)$ with
  $[L(\autA)]=\calL$. Without loss of generality, we can assume that
  $S\cap A=\emptyset$ holds. We want to simulate the runs of $\autA$
  in inverse direction with the help of a (stateless) trace-pushdown
  system $\autP$. To this end, our system initially writes a final
  state $f\in F$ on top of the pushdown of $\autP$. Then we simulate
  an edge $(p,a,q)\in T$ by replacing the state $q$ at the top of the
  pushdown by the word $pa$. We are done, if $\autA$ is in an initial
  state $\iota\in I$.  In this case, we simply pop the state $\iota$
  from the pushdown.
  
  For this construction we first have to extend our dependence
  alphabet $\structD$ to $\structD'=(A',D')$ which is also able to
  handle states of $\autA$:
  \begin{itemize}
    \item $A':=A\cup S\cup\{\#\}$ where $\#\notin A\cup S$ and
    \item $D'$ is the symmetric closure of
      $D\cup ((S\cup\{\#\})\times A')$, i.e., the new letters from
      $S\cup\{\#\}$ are dependent from all letters in $A'$.
  \end{itemize}
  Then we construct a trace-pushdown system
  $\autP=(Q,\Delta')$ with
  \begin{itemize}
    \item $Q:=\{\top\}$ and
    \item $\Delta'$ consists of the following transitions:
    \begin{itemize}
      \item $(\top,\#,f,\top)$ for each $f\in F$,
      \item $(\top,p,qa,\top)$ for each transition
        $(q,a,p)\in\Delta$, and
      \item $(\top,\iota,\varepsilon,\top)$ for each
        $\iota\in I$.
    \end{itemize}
  \end{itemize}
  As mentioned before, $\autP$ simulates computations of $\autA$
  backwards. Underneath the top symbol $q\in S$ of our pushdown, we
  find a trace $[w]\in\traceMon$ such that $\autA$ has a
  $v$-labeled run (with $v\sim w$) from $q$ to
  $F$. Formally, for all $p,q\in Q$ and $w\in A^*$, we have
  $(\top,[p])\vdash_\autP^*(\top,[qw])$ if, and only
  if, there is a word $v\in A^*$ with $q\xrightarrow{v}_\autA p$ and
  $v\sim w$. This implies
  \[
    (\top,[\#]) \vdash_\autP^* (\top,[w])
    \iff
    [w]\in [L(\autA)]
  \]
  for each word $w\in A^*$. Hence,
  $\calL=[L(\autA)]=\{[w]\mid
  (\top,[\#])\vdash_\autP^*(\top,[w])\}$.
\end{proof}

\section{Conclusion}

In this paper, we investigate to what extend Caucal's preservation
results for pushdown systems can be extended to trace-pushdown
systems. Under a natural restriction of the transition relation, this
turns out to be possible, but requires nontrivial extensions of his
ideas: rational relations have to be restricted to lc-rational
relations (whose theory is developed in this paper), the two distinct
generalisations of regular languages to trace languages play different
roles, and the saturation of a system does not allow to replace every
path by a path with at most two phases, but only by a path with a
bounded number of phases where the bound depends on the dependence
alphabet.

We show that the reachability relation is decidable in polynomial
time. This result complements Zetzsche's work on valence automata over
loop-free graph monoids: while he imposes conditions on the monoid, we restrict
the transition relation.

While the reachability problem is decidable, it is still an open question
whether also the recurrent and alternating reachability problems for
trace-pushdown systems are decidable. Also the expressive power of
trace-pushdown systems is still unknown.

\bibliography{lit}

% Count number of todos
\newoutputstream{todos}
\openoutputfile{paper.todos.ctr}{todos}
\addtostream{todos}{\arabic{@todonotes@numberoftodonotes}}
\closeoutputstream{todos}

\end{document}